%% file: main.tex
\documentclass[11pt]{article}
\usepackage[margin=1in]{geometry}
\usepackage{amsmath,amsthm,amsfonts,amssymb}
\usepackage[colorlinks=true, urlcolor=blue, linkcolor=blue, citecolor=magenta]{hyperref}
\usepackage{enumitem}
\usepackage[T1]{fontenc}
\usepackage{aliascnt}
\usepackage{graphicx}
\usepackage{tikz}
\usetikzlibrary{arrows.meta}

\usepackage{booktabs,tabularx,array}
\usepackage{caption}
\usepackage{bbm}

\allowdisplaybreaks

\usepackage[capitalise,nameinlink]{cleveref}
\Crefname{theorem}{Theorem}{Theorems}
\Crefname{subsection}{Subsection}{Subsections}
\Crefname{section}{Section}{Sections}
\Crefformat{equation}{(#2#1#3)}
\crefname{subsection}{Subsec.}{Subsecs.}
\Crefname{subsection}{Subsec.}{Subsecs.}

\newtheorem{theorem}{Theorem}[section]

\newcommand{\newaliastheorem}[2]{%
  \newaliascnt{#1}{theorem}%
  \newtheorem{#1}[#1]{#2}%
  \aliascntresetthe{#1}%
  \crefname{#1}{#2}{#2s}%
  \Crefname{#1}{#2}{#2s}%
}

\newaliastheorem{proposition}{Proposition}
\newaliastheorem{lemma}{Lemma}

\theoremstyle{definition}
\newaliastheorem{definition}{Definition}
\newaliastheorem{remark}{Remark}

\usepackage{thmtools}
\usepackage{thm-restate}

\DeclareMathOperator*{\E}{\mathbb E}

\newcommand{\floorbra}[1]{\left\lfloor #1 \right\rfloor}
\newcommand{\ceilbra}[1]{\left\lceil #1 \right\rceil}

\newcommand{\Bcal}{\mathcal{B}}
\newcommand{\Fcal}{\mathcal{F}}
\newcommand{\Ical}{\mathcal{I}}
\newcommand{\Lcal}{\mathcal{L}}
\newcommand{\Mcal}{\mathcal{M}}
\newcommand{\Qcal}{\mathcal{Q}}
\newcommand{\Ucal}{\mathcal{U}}
\newcommand{\Vcal}{\mathcal{V}}

\newcommand{\F}{\mathbb{F}}

\title{Strong Inapproximability for a Promise Rank Problem} 
\author{
Venkatesan Guruswami
\thanks{Simons Institute for the Theory of Computing, and Departments of ECCS \& Mathematics, UC Berkeley. Email: {\tt venkatg@berkeley.edu}. Research supported in part by NSF grant CCF-2211972, a DARPA grant under Contract No. HR0011262E031, and a Simons Investigator award.}
\and 
Xuandi Ren\thanks{Department of EECS, UC Berkeley. Email: \texttt{xuandi\_ren@berkeley.edu}. Supported in part by NSF grant CCF-2228287 and V.G's Simons Investigator award.}
\and 
Shaoxuan Tang\thanks{Institute for Interdisciplinary Information Sciences, Tsinghua University. Some of this work was done when visiting UC Berkeley. Email: \texttt{tsx23@mails.tsinghua.edu.cn}.}
}
\date{}

\begin{document}

\maketitle
\thispagestyle{empty}

\begin{abstract}

Given a linear subspace of $n \times n$ matrices over $\F_{2^r}$ that is promised to contain a matrix of rank $1$, we prove that it is hard to find a matrix of rank $n^{o(1/\log \log n)}$, assuming NP doesn't have sub-exponential algorithms. In addition to being a basic problem, the hardness of this problem, even for the exact version, drove recent PCP-free inapproximability results for minimum distance and shortest vector problems concerning codes and lattices.

\smallskip
The proof combines the concept of superposition soundness introduced by Khot and Saket with moment matrices. To produce a rank-gap of $1$ vs. $k$, the reduction runs in time $n^{O(\log k)}$. We also give another moment-matrix-based construction which runs in time $n^{O(k)}$ but works for any finite field $\mathbb F_q$.

\end{abstract}

\newpage
\tableofcontents
\thispagestyle{empty}

\clearpage

\input{sec/intro}
\input{sec/pre}

\input{sec/f2}
\input{sec/moment}

\section*{Acknowledgement}
 The assistance of ChatGPT was used with some of the proofs in \Cref{sec:moment-matrix-construction}. All proofs were carefully verified by the authors.

\bibliographystyle{alpha}
\bibliography{ref}

\appendix
\input{sec/inapprox_corollaries}
\input{sec/proof_Khot_Saket}

\end{document}

%% file: sec/intro.tex
\section{Introduction}
\label{sec:introduction}

Given a linear subspace $\Lcal\subseteq \F^{N\times N}$ that is promised to
contain a matrix of rank one, we study the problem of finding a nonzero matrix in
$\Lcal$ of minimum rank.

This problem is closely related to the problem of finding the
\textit{minimum rank distance} of a rank-metric code
\cite{Courtois01,GaboritZemor16,Ravagnani16}.  In fact, there is a simple and
classical embedding of Hamming-metric codes into this setting.  Given a linear
code $C\subseteq \F^N$, consider the diagonal matrix code
\[
    \operatorname{Diag}(C):=\{\operatorname{diag}(c):c\in C\}
    \subseteq \F^{N\times N}.
\]
For every $c\in C$,
\[
    \operatorname{rank}(\operatorname{diag}(c))=\|c\|_0,
\]
and hence the minimum rank distance of $\operatorname{Diag}(C)$ is exactly the
minimum Hamming distance of $C$.  Consequently, hardness results for the
Minimum Distance Problem (MDP) \cite{DMS03,CW12,AK14,Mic14,BGLR25} directly
translate to hardness for approximating the minimum rank in a matrix subspace.

However, this diagonal reduction also explains what it does \emph{not} prove.
The YES case inherited from MDP only promises a nonzero matrix of rank $d(C)$,
the minimum distance of the code, which is generally not $1$.  In fact, it is trivial to check if the minimum distance of a code is 1.
It is therefore
natural to ask whether the problem remains hard under the stronger promise that
the subspace contains a nonzero matrix of smallest possible rank, namely rank $1$.

A second motivation comes from recent PCP-free inapproximability results for
sparse vector problems, including the Minimum Distance Problem and the Shortest
Vector Problem \cite{BGLR25}.  These reductions start from systems of quadratic
equations and view each quadratic constraint as a linear constraint on a matrix:
if $X=xx^\top$, then a quadratic form in $x$ becomes a linear form in the
entries of $X$.  Hence honest solutions to the original quadratic system
correspond to rank-one matrices.  The main soundness issue is to rule out
spurious higher-rank matrices.

In \cite{BGLR25}, this issue is handled via a non-overlap lemma.  Specifically,
if $C$ is a linear code of distance $d(C)$ over $\mathbb F$, then honest
rank-one matrices in the tensor code $C\otimes C$ attain the minimum Hamming
weight $d(C)^2$.  On the other hand, every matrix in $C\otimes C$ of rank at
least $2$ has Hamming weight at least $\alpha\cdot d(C)^2$, for some
$\alpha>1$ depending on $\mathbb F$.  Thus the Hamming-weight gap driving the MDP
hardness is, at its core, a structural gap between rank-one matrices and
matrices of rank at least two.  This suggests stripping away the Hamming-weight
objective and asking for the rank gap directly.

In this work, we show strong hardness of distinguishing
\[
    \exists\,0\ne A\in\Lcal \text{ with } \operatorname{rank}(A)=1
    \qquad\text{from}\qquad
    \forall\,0\ne A\in\Lcal,\ \operatorname{rank}(A)>k
\]
over any fixed finite field of characteristic 2.  Specifically, we give a deterministic reduction from
\textsc{3Sat} producing a rank gap of $1$ versus $k$, with matrix dimension
$n^{O(\log k)}$.  The resulting inapproximability consequences are summarized
in the following main theorem.

\begin{restatable}{theorem}{maininapproximability}
\label{thm:intro-main-inapproximability}
Let
\[
    \operatorname{minrank}(\Lcal)
    :=\min\{\operatorname{rank}(A):0\ne A\in\Lcal\}.
\]
For every fixed integer $r \ge 1$, no polynomial-time algorithm can, given a linear subspace
$\Lcal\subseteq\F_{2^r}^{N\times N}$, distinguish between
\begin{itemize}[leftmargin=2em,itemsep=0.15ex,topsep=0.35ex,parsep=0pt]
    \item (YES) $\operatorname{minrank}(\Lcal)=1$,
    \item (NO) $\operatorname{minrank}(\Lcal)>\gamma$,
\end{itemize}
in the following regimes:
\begin{enumerate}[label=\textup{(\alph*)},leftmargin=2em,itemsep=0.25ex,topsep=0.35ex,parsep=0pt]
    \item assuming $\mathsf{NP}\ne\mathsf{P}$, when $\gamma>1$ is any constant;
    \item assuming
    $\mathsf{NP}\nsubseteq \mathsf{DTIME}(2^{\log^{O(1)} n})$, when
    $\gamma=2^{(\log N)^{1-\epsilon}}$ for any fixed $0<\epsilon<1$;
    \item assuming
    $\mathsf{NP}\nsubseteq \bigcap_{\delta>0}\mathsf{DTIME}(2^{n^\delta})$, when
    $\gamma=N^{c/\log\log N}$ for some fixed constant $c>0$.
\end{enumerate}
\end{restatable}

We present the reduction in \Cref{sec:rank-one-vs-low-rank}. The proof combines the superposition-soundness framework of Khot and Saket
\cite{KhotSaket2014} with linearized moment matrices.  A satisfying assignment
to the starting quadratic system gives a rank-one moment matrix.  Conversely, a
low-rank feasible matrix can be decomposed into a bounded number of symmetric
rank-one pieces; these pieces behave like several assignments satisfying the
quadratic system in superposition.  The soundness analysis of Khot--Saket soundness rules out such
superpositions, while the moment-matrix equal-union constraints ensure that the
remaining zero-sum case cannot hide a nonzero low-rank matrix.

\medskip
We also include a simpler direct
moment-matrix reduction in \Cref{sec:moment-matrix-construction}. This
construction avoids the superposition-soundness and works over every finite
field $\F_q$, at the cost of producing matrices of dimension $n^{O(k)}$.  The
corresponding inapproximability statement is the following.

\begin{restatable}{theorem}{anyfieldinapproximability}
\label{thm:intro-any-field-inapproximability}
Let
\[
    \operatorname{minrank}(\Lcal)
    :=\min\{\operatorname{rank}(A):0\ne A\in\Lcal\}.
\]
For every fixed finite field $\F_q$, no polynomial-time algorithm can, given a linear subspace
$\Lcal\subseteq\F_q^{N\times N}$, distinguish between
\begin{itemize}[leftmargin=2em,itemsep=0.15ex,topsep=0.35ex,parsep=0pt]
    \item (YES) $\operatorname{minrank}(\Lcal)=1$,
    \item (NO) $\operatorname{minrank}(\Lcal)>\gamma$,
\end{itemize}
in the following regimes:
\begin{enumerate}[label=\textup{(\alph*)},leftmargin=2em,itemsep=0.25ex,topsep=0.35ex,parsep=0pt]
    \item assuming $\mathsf{NP}\ne\mathsf{P}$, when $\gamma>1$ is any constant;
    \item assuming
    $\mathsf{NP}\nsubseteq \mathsf{DTIME}(2^{\log^{O(1)} n})$, when
    $\gamma=(\log N)^{1-\epsilon}$ for any fixed $0<\epsilon<1$;
    \item assuming
    $\mathsf{NP}\nsubseteq \bigcap_{\delta>0}\mathsf{DTIME}(2^{n^\delta})$, when
    $\gamma=c\log N/\log\log N$ for some fixed constant $c>0$.
\end{enumerate}
\end{restatable}

\subsection{Proof overview}
\label{subsec:rank-proof-overview}

The key ingredient in both reductions is the pseudo-moment matrix.  Fix a
moment level $d$. We introduce a pseudo-moment coordinate $y_R$ for every Boolean monomial $x^R:=\prod_{i\in R}x_i$ of degree at most $2d$. Note that the vector $(y_R)$ is not assumed to come from an actual Boolean assignment, for which reason it is called pseudo-moment vector. The associated degree-$d$ pseudo-moment matrix\footnote{This is the standard moment-matrix object from the Sum-of-Squares literature: for a pseudoexpectation
$\widetilde{\mathbb E}$ over the Boolean cube, the moment
matrix has entries
$\widetilde{\mathbb E}[x^Sx^T]=\widetilde{\mathbb E}[x^{S\cup T}]$.
The Lasserre/SoS hierarchy additionally imposes normalization and
positive-semidefiniteness of moment matrix.  See,
e.g., \cite{Lasserre01,Parrilo03,Laurent09,BarakSteurer14}.} is defined by
\[
        H_d(y)_{S,T}=y_{S\cup T}.\qquad (|S|,|T|\leq d)
\]

In the matrix formulation, we impose the equal-union constraints
\[
        A_{S,T}=A_{S',T'}
        \qquad\text{whenever}\qquad
        S\cup T=S'\cup T'.
\]
This way, the matrix is not an arbitrary linearization of the products
$x^Sx^T$.  All factorizations of the same Boolean monomial $x^{S\cup T}$ are
forced to share a single pseudo-moment coordinate.

This hidden redundancy is what makes rank a useful test.  If a pseudo-moment $y_R$ is
nonzero, it is not isolated in one entry of the matrix: it appears in every
entry whose row and column labels union to $R$.  Therefore a feasible low-rank
matrix is constrained not only by the original equations, but also by the
truncated Boolean monomial algebra encoded by these equal-union identities.

We first record the simple but powerful observation related to the rank of pseudo-moment matrices that is used in both reductions, and then give an 
overview of the two reductions separately.

\begin{lemma}[Informal structural lemma, used in \Cref{lem:zero-resultant-rank-v3} and \Cref{lem:elementary-flat-level}]
\label{lemma:informal_rank}
    Let $A$ be a matrix satisfying the above equal-union constraints, and let
    $y$ be the pseudo-moment vector it represents. Suppose $R$ is a
    minimum-size set with $y_R\ne0$. Then,
    \[
        \operatorname{rank}(A)
        \ge \binom{|R|}{\floorbra{|R|/2}} .
    \]
\end{lemma}

The proof idea is to look at the submatrix with rows indexed by the
$\floorbra{|R|/2}$-subsets $F\subseteq R$ and columns indexed by the
$\ceilbra{|R|/2}$-subsets $G\subseteq R$.  Its $(F,G)$ entry is
$y_{F\cup G}$.  By the minimality of $R$, this entry is zero unless
$F\cup G=R$.  Hence the
submatrix is a permutation matrix, giving the rank
lower bound.

The two reductions use this observation in different ways.  In the
$n^{O(\log k)}$ superposition reduction, the observation turns the zero-sum case from
Khot--Saket soundness into the conclusion that the whole low-rank matrix is
zero.  In the $n^{O(k)}$ direct pseudo-moment reduction, the same observation forces the existence of a
nonzero flat level from which one can round to an honest Boolean solution.

\paragraph{The $n^{O(\log k)}$ superposition reduction over $\mathbb F_{2^r}$.}

We first prove the rank gap over $\mathbb F_2$ and then extend it to
$\mathbb F_{2^r}$ in a black-box way.  We adapt the Khot--Saket superposition
reduction \cite{KhotSaket2014} to produce a constant-free quadratic system
$\Qcal$.  After introducing a homogenizing variable $x_0$, and  introducing variables $y_S$ for non-constant monomials $x^S$ of degree at most
$d=\Theta(\log k)$,  we obtain a quadratic system with the following property.  In the YES case, an honest
satisfying assignment gives a solution to $\Qcal$.  In the NO case, if $t$
assignments satisfy $\Qcal$ in superposition, meaning that the sum of their
evaluations on every equation of $\Qcal$ is zero, then their coordinate-wise sum
is the zero assignment. 

After that, for each quadratic equation
\[
    \sum_{S,T}c_{S,T}y_Sy_T+\sum_R b_Ry_R=0
\]
in $\Qcal$, we linearize it as
\[
    \sum_{S,T}c_{S,T}A_{S,T}+\sum_R b_RA_{R,R}=0,
\]
and we also impose all equal-union constraints.  

For completeness, an honest satisfying assignment gives a feasible rank-one
moment matrix.

For soundness, suppose in the NO case that a nonzero feasible matrix $A$ has
rank at most $k$.  A decomposition lemma\footnote{This decomposition is not obvious: we require a sum of \textit{symmetric} rank-$1$ terms, which is more restrictive than a general rank decomposition. 
} shows 
\[
    A=\sum_{i=1}^t u^{(i)}u^{(i)\top},
    \qquad t\le \left\lfloor\frac{3k}{2}\right\rfloor .
\]
Substituting this decomposition
into the linearized constraints shows that the
assignments $u^{(i)}$ satisfy $\Qcal$ in superposition\footnote{The constant-free assumption is needed here.  If a quadratic equation has a nonzero constant term, substituting a decomposition
\(A=\sum_{i=1}^t u^{(i)}u^{(i)\top}\) counts that constant with a parity depending on
the number of summands, so a linearized feasible matrix yields the desired
superposition condition only when $t$ is odd. Huang~\cite{Huang15},
building on the Khot--Saket superposition framework, uses precisely this
oddness phenomenon: he proves that superposition satisfaction by an odd number
of assignments is equivalent to odd-covering.}.  By the soundness
guarantee, the sum of these assignments is zero.  Consequently, every diagonal
entry of $A$ vanishes:
\[
    A_{S,S}=\sum_i (u^{(i)}_S)^2=\sum_i u^{(i)}_S=0.
\]
The equal-union constraints then imply
$A_{S,T}=A_{S\cup T,S\cup T}=0$ whenever $|S\cup T|\le d$.

We still need to rule out nonzero entries whose union has size larger than
$d$.  Suppose such an entry exists, and choose $A_{S,T}\ne0$ with
$|S\cup T|$ minimum.  Put $R=S\cup T$.  By the structural lemma,
\[
    \operatorname{rank}(A)
    \ge \binom{|R|}{\floorbra{|R|/2}}
    \ge \binom{d+1}{\floorbra{(d+1)/2}}>k,
\]
contradicting the rank assumption.  

Finally, to work over $\mathbb F_{2^r}$, we interpret the same
$\mathbb F_2$-linear equations over the extension field; a nonzero
rank-at-most-$k$ matrix over $\mathbb F_{2^r}$ would descend, via an
$\mathbb F_2$-linear functional, to a nonzero matrix over $\mathbb F_2$ of rank
at most $rk$. So we just run the base-field reduction with gap rank parameter $rk$.

\medskip\noindent\textbf{The $n^{O(k)}$ direct reduction.}
This construction starts from \textup{\textsc{QuadEq}} over an arbitrary finite
field $\F_q$.  We set $d=k$ and build the pseudo-moment matrix whose rows and
columns are indexed by all monomials of degree at most $d$, including the empty
monomial.  We impose the same equal-union constraints, and for each 
equation of the \textup{\textsc{QuadEq}} instance
\[
    f_\ell(x)=\sum_{\substack{U\subseteq[n], \ |U|\le2}} c_{\ell,U}x^U
\]
we impose the following pseudo-moment versions of the identities $x^W f_\ell(x)=0$:
\[
    \sum_{\substack{U\subseteq[n]\\ |U|\le2}} c_{\ell,U}y_{U\cup W}=0
    \qquad \forall W, \ |W|\le 2d-2 \ .
\]
These are the pseudo-moment versions of the identities $x^W f_\ell(x)=0$.
The completeness is immediate: a Boolean solution gives a rank-one moment
matrix solution, and it is nonzero because the empty-coordinate entry is $1$.

For soundness, suppose the source instance is unsatisfiable but there is a
nonzero feasible pseudo-moment matrix $H_d(y)$ of rank at most $d$.  Let
$r_e=\operatorname{rank}H_e(y)$; this sequence is nondecreasing in $e$.  The
structural lemma shows that the first nonzero level already has a reasonably large rank that the sequence
$r_e$ cannot keep increasing strictly up to level $d$.  Hence there must be a
nonzero flat level
\[
    r_e=r_{e+1}>0
\]
for some $e<d$.

At such a flat level, multiplying a column label by $x_i$ creates no new column
direction.  Thus multiplication by $x_i$ defines a linear operator $T_i$ on the
column space $C_e$ of $H_e(y)$.  These operators satisfy the Boolean rules
\[
    T_i^2=T_i,
    \qquad
    T_iT_j=T_jT_i,
\]
and the localizing constraints translate into the operator identities
$f_\ell(T)=0$ on $C_e$ for every source equation $f_\ell$.  Since the $T_i$ are
commuting projections over $\F_q$, they have a common eigenvector with
eigenvalues $a_i\in\{0,1\}$.  The tuple $a=(a_1,\ldots,a_n)$ is therefore a
Boolean point, and the identities $f_\ell(T)=0$ imply $f_\ell(a)=0$ for all
$\ell$, contradicting unsatisfiability.

%% file: sec/pre.tex
\section{Preliminaries}
\label{sec:preliminaries}
We now formally define the concepts related to monomial-indexing and pseudo-moments used
in \Cref{sec:rank-one-vs-low-rank} and \Cref{sec:moment-matrix-construction}, which we already sketched in the introduction. 

Let $[n]:=\{1,2,\dots,n\}$ for any positive integer $n$. For a subset $S$, we write $x^S :=\prod_{i \in S} x_i$, and we set $x^\emptyset:=1$. All products of monomials and polynomials are taken in the Boolean
squarefree monomial algebra. Equivalently, over the relevant field we
identify polynomials modulo the relations \(x_i^2=x_i\), so that
\(x^Sx^T=x^{S\cup T}\).

Write 
\[
\begin{aligned}
    \Ucal_{n,d}:=\{S\subseteq \{0,1,\dots,n\}: 1\le |S|\le d\};\\
    \Vcal_{n,d}:=\{S\subseteq \{1,\ldots,n\}: 0 \le |S|\le d\}.
\end{aligned}
\]
Here \(\Ucal_{n,d}\) is used for the homogenized $n^{O(\log k)}$ construction:
it allows the coordinate \(0\) but excludes the empty monomial. The set
\(\Vcal_{n,d}\) is used for the direct $n^{O(k)}$ construction: it uses only the original variables and includes the empty monomial.

When $\mathcal W_d$ is either $\Ucal_{n,d}$ or $\Vcal_{n,d}$, define a \emph{monomial assignment of degree $d$} to be a map
\[
    \sigma:\{x^S:S\in\mathcal W_{d}\}\to \F.
\]
Equivalently, any monomial assignment of degree $d$ can be written as a \emph{degree-$d$ pseudo-moment assignment}, or \emph{pseudo-moment vector}, 
\[
    y=(y_R)_{R\in\mathcal W_{d}}.
\]

When a degree-$2d$ pseudo-moment vector is given, its associated \textit{pseudo-moment matrix} is defined as
\[
    H_d(y):=\bigl(y_{S\cup T}\bigr)_{S,T\in\mathcal W_d}.
\]
If $a$ is an actual Boolean point in $\{0,1\}^{n+1}$ (or $\{0,1\}^n$, respectively), the ``honest'' moment
assignment generated by $a$ is $y_R=a^R$.  The corresponding moment vector is
\[
    v_d(a):=(a^S)_{S\in\mathcal W_d},
\]
and in this case
\[
    H_d(y)=v_d(a)v_d(a)^\top.
\]
In particular, $H_d(y)$ has rank one whenever $v_d(a)\ne 0$.

On the other hand, any matrix $A \in \mathbb F^{\mathcal W_d \times \mathcal W_d}$ satisfying
\[
    A_{S,T}=A_{S',T'}
    \qquad\text{whenever }S\cup T=S'\cup T',
\]
is a pseudo-moment matrix $H_d(y)$ for some unique pseudo-moment vector $y\in\F^{\mathcal W_{2d}}$.  We therefore sometimes state the linear constraints on the pseudo-moment matrix in terms of the associated pseudo-moment vector $y$.

%% file: sec/f2.tex
\section{An $n^{O(\log k)}$ Construction over $\mathbb F_{2^r}$}
\label{sec:rank-one-vs-low-rank}

In this section, we prove the following theorem.

\begin{theorem}
\label{thm:rank-gap-subspace-v3}
For every fixed integer $r \ge 1$ and every integer $k\ge 1$, there is a deterministic $n^{O(\log k)}$-time reduction which maps a
\textup{\textsc{3Sat}} instance on $n$ variables to a linear subspace
\[
    \Lcal\subseteq \F_{2^r}^{N\times N},
    \qquad N=n^{O(\log k)},
\]
such that:
\begin{itemize}
    \item (YES) if the input formula is satisfiable, then $\Lcal$ contains a nonzero
    matrix of rank~$1$;
    \item (NO) if the input formula is unsatisfiable, then $\Lcal$ contains no
    nonzero matrix of rank at most $k$.
\end{itemize}
\end{theorem}

To prove \Cref{thm:rank-gap-subspace-v3}, we first prove the result over $\mathbb F_2$, then transfer it to any $\mathbb F_{2^r}$ by a rank-descent argument. Our reduction has four steps:
\begin{enumerate}
    \item First, we reduce \textsc{3Sat} to a degree-$d$ polynomial equation system $\Bcal$ over $\F_2$ on the variables $x_0,x_1,\ldots,x_n$. The homogenizing variable $x_0$ replaces the constant 1 and  makes every equation in $\Bcal$ have zero constant term.
    \item Next, we replace each nonconstant monomial $x^S$, $S\in\Ucal_{n,d}$, by a new
    variable $y_S$, and we add the moment constraints
    $y_Sy_T=y_{S\cup T}$. This produces a quadratic equation system $\Qcal$ with zero constant term. Furthermore, we obtain Khot--Saket superposition soundness: in the NO case, if $t$ assignments
    $\tau^{(1)},\ldots,\tau^{(t)}$ satisfy every equation of $\Qcal$ in superposition (see \Cref{def:superposition-v3}), then the sum of the $t$ assignments
    \[
        \tau:=\sum_{i=1}^t \tau^{(i)}
    \]
    must vanish on every $y$-coordinate.
    \item We then linearize the quadratic system $\Qcal$ by introducing matrix variables
    $A_{S,T}$ for the products $y_Sy_T$. By adding equal-union constraints to enforce the Boolean monomial identity $x^Sx^T=x^{S\cup T}$, and by the rank observation in \Cref{lemma:informal_rank}, we get the rank gap over $\mathbb F_2$.
    \item Finally, we extend the hardness from $\F_2$ to $\F_{2^r}$ by interpreting
    the same homogeneous linear constraints over the extension field. The rank-descent lemma
    \Cref{lem:rank-descent-extension-v3} converts any nonzero matrix of rank at most $k$ over $\F_{2^r}$ into a nonzero matrix over $\F_2$ of rank at most
    $rk$.  Running the $\F_2$ construction with rank parameter $rk$ gives the
    claimed gap over $\F_{2^r}$.
\end{enumerate}

The four steps are carried out in
\Cref{subsec:z-to-x-constant-free,subsec:x-to-y-superposition,subsec:linearized-superposition-instance,subsec:f2tof2r},
respectively. For orientation, the first three variable-changing steps are summarized in
\Cref{fig:section3-roadmap}; the corresponding variables and their roles are
listed in \Cref{tab:section3-variable-layers}.

\begin{center}
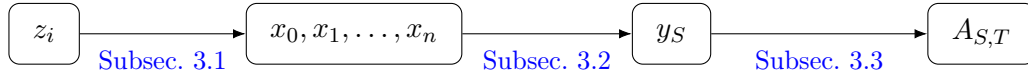

\begin{tikzpicture}[
    >=Latex,
    layer/.style={
        draw,
        rounded corners,
        align=center,
        inner xsep=9pt,
        inner ysep=5pt,
        minimum height=2.1em,
        font=\normalsize
    },
    arrowlabel/.style={
        font=\small,
        align=center,
        text width=3.1cm
    }
]
\node[layer] (z) at (0,0) {$z_i$};
\node[layer] (x) at (4.1,0) {$x_0,x_1,\ldots,x_n$};
\node[layer] (y) at (8.3,0) {$y_S$};
\node[layer] (A) at (12.4,0) {$A_{S,T}$};
\draw[->] (z) -- node[midway, below=0.25em, arrowlabel]
    {\Cref{subsec:z-to-x-constant-free}} (x);

\draw[->] (x) -- node[midway, below=0.25em, arrowlabel]
    {\Cref{subsec:x-to-y-superposition}} (y);

\draw[->] (y) -- node[midway, below=0.25em, arrowlabel]
    {\Cref{subsec:linearized-superposition-instance}} (A);
\end{tikzpicture}
\captionof{figure}{Roadmap of the reduction steps in this section.}
\label{fig:section3-roadmap}
\end{center}

\begin{center}
\renewcommand{\arraystretch}{1.25}
\begin{tabularx}{\linewidth}{
    @{}
    >{\centering\arraybackslash}p{0.20\linewidth}
    >{\centering\arraybackslash}p{0.32\linewidth}
    >{\centering\arraybackslash}X
    @{}
}
\toprule
\textbf{Variables} & \textbf{Meaning} & \textbf{Used in} \\
\midrule

$z_i$
&
original \textsc{3Sat} variables
&
starting NP-hardness source
\\

$x_0,x_1,\ldots,x_n$
&
homogenized Boolean variables
&
degree-$d$ polynomial system $\Bcal$ with no constant term; $x_0$ replaces $1$
\\

$y_S$
&
variable for monomial $x^S$
&
quadratic system $\Qcal$ with Khot--Saket superposition soundness
\\

$A_{S,T}$
&
matrix variable linearizing $y_Sy_T$
&
linear matrix subspace $\Lcal_d(\Qcal)$ for the rank-gap reduction
\\

\bottomrule
\end{tabularx}
\captionof{table}{Variables defined in this section and their roles.}
\label{tab:section3-variable-layers}
\end{center}

\subsection{Reducing from \textsc{3Sat} to constant-free polynomial equations}
\label{subsec:z-to-x-constant-free}

In this subsection, before presenting our main reduction in \Cref{thm:homogeneous-low-weight-v3}, we first introduce the following duality fact used in \cite{KhotSaket2014}: any
assignment to all nonconstant monomials of degree at most $d$ can be represented as the
sum of evaluations at actual points of $\F_2^{n+1}$. 
\label{subsec:homogeneous-ks-source}
\begin{restatable}[Monomial assignments as sums of points,~\cite{KhotSaket2014}]{lemma}{monomialassignmentsaspoints}
\label{lem:monomial-assignments-as-points-v3}
Let $\F_2[x]_{\le d}$ denote the $\F_2$-vector space of multilinear polynomials in variables $x_0, \ldots, x_n$ of degree at most $d$. For every monomial assignment 
\[
\sigma : \{x^S : S \in \mathcal U_{n,d}\} \to \mathbb F_2,
\]
define $\sigma(1) = 1$ and extend $\sigma$ by linearity to all polynomials in $\F_2[x]_{\le d}$. Then, there is a subset $\beta\subseteq \mathbb F_2^{n+1}$ such that
\[
\sigma(q)=\sum_{a\in\beta}q(a)
\]
holds for every $q \in \F_2[x]_{\le d}$.
\end{restatable}

Since $\sigma(1)=1$, any set
\(\beta\) representing \(\sigma\) must have odd cardinality. Among all sets representing \(\sigma\) in this sense, fix one with minimum cardinality and denote it by \(\beta_\sigma\). 
We put a proof of \Cref{lem:monomial-assignments-as-points-v3} in \Cref{sec:proof-khot-saket}.

The following point-isolator lemma is also useful in proving \Cref{thm:homogeneous-low-weight-v3}.

\begin{lemma}[Point isolator]
\label{lem:point-isolator-v3}
Let $\rho\ge 0$ be an integer, $T\subseteq \F_2^{n+1}$ with $|T|<2^\rho$, and let $a\in T$. Then there exists a
multilinear polynomial $q$ of degree at most $\rho$ such that
\[
    q(a)=1\qquad\text{and}\qquad q(b)=0\quad\text{for all }b\in T\setminus\{a\}.
\]
\end{lemma}

\begin{proof} 
Let $m = n+1$. The set of evaluations of all degree-$\le \rho$ multilinear polynomials over $\F_2^m$ forms the classical Reed--Muller code $\text{RM}(\rho, m)$. We consider the evaluation map $E: \text{RM}(\rho, m) \to \F_2^T$ that restricts these polynomials to the points in $T$. To find the desired polynomial $q$, it suffices to show the evaluation map $E$ is surjective.

By basic linear algebra, a linear map is surjective if and only if the orthogonal complement of its image is trivial. Suppose $g \in \F_2^T$ is a vector in this orthogonal complement. This means $g$ satisfies
\[
    \sum_{v \in T} g(v) q(v) = 0
\]
for every polynomial $q$ of degree at most $\rho$. We can extend $g$ to a vector over the entire space $\F_2^m$ by setting $g(v) = 0$ for all $v \notin T$. The above equation then states that $g$ is orthogonal to every codeword in $\text{RM}(\rho, m)$, therefore $g$ is in the dual code $\text{RM}(\rho, m)^\perp$. 

It is a standard fact from coding theory (e.g., \cite{MacWilliamsSloane1977}) that the dual of $\text{RM}(\rho, m)$ is $\text{RM}(m-\rho-1, m)$, and its minimum Hamming distance is $2^{m - (m-\rho-1)} = 2^{\rho+1}$. However, any dual codeword supported on $T$ would have Hamming weight at most $|T| < 2^\rho$, which contradicts the minimum distance bound. 

Therefore, no such non-zero dual codeword exists. The evaluation map is indeed surjective, implying we can interpolate any function on $T$, including the indicator function of the point $a$.
\end{proof}

We now give the reduction from \textsc{3Sat} to a system of degree-$d$ constant-free polynomial equations. This is a constant-free analogue of the Khot--Saket construction. 
We homogenize the usual Boolean constraints by introducing a variable $x_0$, which plays the role of the constant $1$ in honest
assignments, and by adding the equations $x_i(x_i+x_0)=0$.  If $x_0=1$, the
clause polynomials encode the original \textup{\textsc{3Sat}} instance.  If
$x_0=0$, these equations force every variable to vanish as well. 
\begin{theorem}[Constant-free low-weight soundness]
\label{thm:homogeneous-low-weight-v3}
For every integer $d\ge 4$,  
there is a deterministic $n^{O(d)}$-time
reduction from a \textup{\textsc{3Sat}} instance on $n$ variables to a system $\Bcal$ of constant-free degree-$d$ polynomial equations in $n+1$ variables over $\F_2$ 
such that:
\begin{itemize}
    \item (YES) if the input formula is satisfiable, then $\Bcal$ has a satisfying
    assignment $a\in\F_2^{n+1}$ with $a_0=1$;
    \item (NO) if the input formula is unsatisfiable and $\sigma$ is a monomial
    assignment of degree $d$ satisfying all equations in $\Bcal$, then either
    $\beta_\sigma=\{\mathbf 0\}$ or $|\beta_\sigma|\ge 2^{d-3}$. 
\end{itemize}
\end{theorem}

\begin{proof}
Let the input formula be $\varphi$ with variables $z_1,\ldots,z_n$ and clauses
$C_1,\ldots,C_m$.  We introduce variables $x_0,x_1,\ldots,x_n$.  For a literal
$\ell$, define its false-literal form by
\[
    h_\ell(x)=
    \begin{cases}
        x_0+x_i & \text{if }\ell=z_i,\\
        x_i & \text{if }\ell=\neg z_i.
    \end{cases}
\]
When $x_0=1$, this is $1$ if and only if the literal $\ell$ is false.  For a clause
$C_j=\ell_{j,1}\vee\ell_{j,2}\vee\ell_{j,3}$, define the clause polynomial
\[
    p_j(x):=h_{\ell_{j,1}}(x)h_{\ell_{j,2}}(x)h_{\ell_{j,3}}(x).
\]
For every original variable, define the Booleanity
polynomial
\[
    b_i(x):=x_i(x_i+x_0).
\]
Every $p_j$ and $b_i$ has no constant term,
and their degrees are at most $3$ and $2$, respectively.

The system $\Bcal$ consists of the following equations:
\begin{enumerate}[label=(\arabic*)]
    \item $x^S p_j(x)=0$ for every $j\in[m]$ and every monomial $x^S$,
    including $S=\emptyset$, of degree at most $d-3$; 
    \item $x^S b_i(x)=0$ for every $i\in[n]$ and every monomial $x^S$,
    including $S=\emptyset$, of degree at most $d-2$.
\end{enumerate}
All equations have degree at most $d$ and no constant term.

In the YES case, take a satisfying Boolean assignment $(a_1,\ldots,a_n)$ to
$\varphi$ and set $a_0=1$.  Then each $b_i(a)$ vanishes.  In every clause at
least one literal is true, so at least one of the three corresponding
false-literal forms is zero, and hence every $p_j(a)$ vanishes.  All monomial
multiples in $\Bcal$ therefore vanish at $a$ as well.

Now assume $\varphi$ is unsatisfiable, and let $\sigma$ be a monomial assignment
satisfying all equations in $\Bcal$.  Suppose, for contradiction, that
$\beta_\sigma\ne\{\mathbf 0\}$ and $|\beta_\sigma|<2^{d-3}$.  By the parity
observation above, \(\beta_\sigma\) is nonempty; since it is not
\(\{\mathbf 0\}\), it contains a nonzero point. Choose such a point
\(a\in\beta_\sigma\). We use the following point-isolating lemma to distinguish $a$ from the remaining points in $\beta_\sigma$.

By \Cref{lem:point-isolator-v3}, there exists a polynomial
$q$ of degree at most $d-3$ such that
\[
    q(a)=1\qquad\text{and}\qquad q(b)=0\quad\text{for all }b\in\beta_\sigma\setminus\{a\}.
\]
First we show that $a_0=1$.  For every $i\in[n]$, the polynomial $q(x)b_i(x)$ is
an $\F_2$-linear combination of equations in $\Bcal$, so
\begin{equation}
\label{eq:isolated-booleanity-v3}
    0=\sigma\bigl(q(x)b_i(x)\bigr)
      =\sum_{b\in\beta_\sigma}q(b)b_i(b)
      =a_i(a_i+a_0).
\end{equation}
If $a_0=0$, then \eqref{eq:isolated-booleanity-v3} forces $a_i=0$ for every
$i\in[n]$, contradicting the choice of the nonzero point $a$.  Hence $a_0=1$.

With $a_0=1$, \eqref{eq:isolated-booleanity-v3} says that the point $a$ obeys the intended Booleanity constraints on the original
coordinates, and the values $h_\ell(a)$ are the usual false-literal indicators of
the Boolean assignment $(a_1,\ldots,a_n)$.  Since $\varphi$ is unsatisfiable,
this assignment falsifies some clause, say $C_j$, and therefore $p_j(a)=1$.
Again $q(x)p_j(x)$ is an $\F_2$-linear combination of equations in $\Bcal$, and
hence
\[
    0=\sigma\bigl(q(x)p_j(x)\bigr)
      =\sum_{b\in\beta_\sigma}q(b)p_j(b)
      =q(a)p_j(a)=1,
\]
a contradiction.  Thus either $\beta_\sigma=\{\mathbf 0\}$ or
$|\beta_\sigma|\ge 2^{d-3}$.
\end{proof}

\subsection{Obtaining Khot--Saket superposition soundness}
\label{subsec:x-to-y-superposition}

We now define superposition satisfaction, a notion ruled out by the soundness analysis of Khot--Saket. We are particularly interested in this notion because assignments in superposition satisfaction naturally correspond to the factors in the symmetric rank-one decomposition given by \Cref{lem:symmetric-decomposition-v3}.

\begin{definition}[Superposition satisfaction and aggregate assignment]
\label{def:superposition-v3}
Let
\[
    g(y)=\sum_{U,V} c_{U,V}y_Uy_V+\sum_W d_W y_W
\]
be a polynomial of degree at most $2$ and without constant term over $\F_2$.  We
say that assignments $\tau^{(1)},\dots,\tau^{(t)}$ \emph{satisfy $g(y)=0$ in
superposition} if
\[
    \sum_{i=1}^t g\bigl(\tau^{(i)}\bigr)=0.
\]
Their \emph{aggregate assignment} 
is the coordinate-wise sum
\[
    \tau:=\sum_{i=1}^t \tau^{(i)}.
\]
\end{definition}
For $\beta\subseteq\F_2^{n+1}$ and a polynomial $f$, write
\[
    \chi_\beta(f):=(-1)^{\sum_{a\in\beta} f(a)}.
\]

The following correlation estimate is from \cite[Lemma 3.3]{KhotSaket2014}, which further builds on \cite[Theorem 20]{DG15}. 

\begin{lemma}[Khot--Saket correlation bound; {\cite[Lemma 3.3]{KhotSaket2014}}] 
\label{lem:ks-correlation-bound-v3}
Let $d\ge 8$ be a positive multiple of $4$.  Let \(\sigma\) be a degree-\(d\) monomial assignment, and let
\(\beta_\sigma\) be a minimum-cardinality set representing \(\sigma\)
as in \Cref{lem:monomial-assignments-as-points-v3}. Suppose \(|\beta_\sigma|\ge 2^{d-3}\), and let $\gamma,\alpha\subseteq\F_2^{n+1}$ be arbitrary.
Let $g$ be a uniformly random multilinear polynomial of degree at most $3d/4$,
and let $h$ be a uniformly random multilinear polynomial of degree at most $d/4$
with zero constant term.  Then
\[
    \left|
    \E_{g,h}\bigl[\chi_{\beta_\sigma}(gh)\chi_\gamma(g)\chi_\alpha(h)\bigr]
    \right|
    \le 2^{-2^{d/4-2}+1}.
\]
\end{lemma}

\begin{theorem}[Constant-free exact-superposition soundness, adapted from \cite{KhotSaket2014}]
\label{thm:homogeneous-superposition-v3}
There is a constant $c>0$ such that the following holds.  Let $t\ge 1$, and let
$d$ be a positive multiple of $4$ with $d\ge \max\{8,c\log(t+1)\}$.  From a
\textup{\textsc{3Sat}} instance on $n$ variables, one can
construct in time $n^{O(d)}$, a system $\Qcal$ of quadratic equations with zero constant term, in variables
\[
    \{y_S:S\in\Ucal_{n,d}\}
\]
with the following properties:
\begin{itemize}
    \item (YES) if the \textup{\textsc{3Sat}} instance is satisfiable, then $\Qcal$
    has a satisfying assignment $y$ with $y_{\{0\}}=1$;
    \item (NO) if the \textup{\textsc{3Sat}} instance is unsatisfiable and
    $\tau^{(1)},\dots,\tau^{(t)}$ 
    satisfy all equations of $\Qcal$ in
    superposition, then their aggregate assignment $\tau:=\sum_{i=1}^t \tau^{(i)}$
    vanishes on every coordinate $y_S$, $S\in\Ucal_{n,d}$.
\end{itemize}
Moreover, $\Qcal$ is obtained by replacing each nonconstant monomial $x^S$ in the
system $\Bcal$ from \Cref{thm:homogeneous-low-weight-v3} with the variable
$y_S$, and by adjoining the constraints
\[
    y_Sy_T=y_{S\cup T}
    \qquad\text{for all }S,T\in\Ucal_{n,d}\text{ with }|S\cup T|\le d.
\]
\end{theorem}

\begin{proof}
Apply \Cref{thm:homogeneous-low-weight-v3} to obtain the constant-free system
$\Bcal$ in the variables $x_0,x_1,\ldots,x_n$.  For every
$S\in\Ucal_{n,d}$ introduce a variable $y_S$, intended to represent the monomial
$x^S$.  Since every equation of $\Bcal$ has zero constant term, no variable
corresponding to $x^\emptyset$ is needed.  The system $\Qcal$ contains:
\begin{enumerate}[label=(\roman*), ref=\roman*]
    \item\label{equation_q:a} for every equation $\sum_S c_Sx^S=0$ in $\Bcal$, the linear equation
    $\sum_S c_Sy_S=0$ obtained by replacing each monomial $x^S$ with $y_S$;
    \item\label{equation_q:b} every multiplicativity constraint $y_Sy_T=y_{S\cup T}$ with
    $S,T\in\Ucal_{n,d}$ and $|S\cup T|\le d$.
\end{enumerate}

The YES case follows directly from \Cref{thm:homogeneous-low-weight-v3}.  If
$a$ is the satisfying assignment there, with $a_0=1$, then setting
$y_S:=a^S$ satisfies all linearized equations from $\Bcal$ and all
multiplicativity constraints.  In particular $y_{\{0\}}=1$.

For the NO case, let $\tau^{(1)},\dots,\tau^{(t)}$ be assignments to the
variables $y_S$ satisfying all equations of $\Qcal$ in superposition.  Write
\[
    \tau^{(i)}_S:=\tau^{(i)}(y_S),
    \qquad
    \tau_S:=\sum_{i=1}^t\tau^{(i)}_S
    \qquad(S\in\Ucal_{n,d}).
\]
Thus $\tau=(\tau_S)_{S\in\Ucal_{n,d}}$ is the aggregate $y$-assignment.
Define the associated degree-$d$ monomial assignments $\sigma$ and
$\sigma^{(i)}$ by
\[
    \sigma(x^S):=\tau_S,
    \qquad
    \sigma^{(i)}(x^S):=\tau^{(i)}_S
    \qquad(S\in\Ucal_{n,d}),
\]
extending them to polynomials by the constant-preserving convention from
\Cref{lem:monomial-assignments-as-points-v3}.  Thus $\tau$ lives on the
$y$-coordinates, while $\sigma$ is the corresponding assignment to the
$x$-monomials.

The main idea here is that, since the equations inherited from \(\mathcal B\) are linear and homogeneous, superposition satisfaction forces the aggregate assignment \(\tau=\sum_i \tau^{(i)}\) to satisfy these equations. Viewing the \(y\)-coordinates of \(\tau\) as a monomial assignment \(\sigma\), \Cref{thm:homogeneous-low-weight-v3} then gives a dichotomy: either \(\beta_\sigma=\{\mathbf 0\}\), or \(\beta_\sigma\) is large. The latter possibility is ruled out by converting the superposition satisfaction of product equations \(y_Sy_T=y_{S\cup T}\) into a low-degree correlation identity, and applying the Khot--Saket correlation bound (\Cref{lem:ks-correlation-bound-v3}). 

The equations inherited from $\Bcal$ become linear after the
monomial-to-variable replacement in (\ref{equation_q:a}), 
and they have no constant term.
Therefore superposition satisfaction implies that $\tau$ satisfies every such
linearized equation. Equivalently, $\sigma$ satisfies every equation of
$\Bcal$.  By \Cref{thm:homogeneous-low-weight-v3}, either
$\beta_\sigma=\{\mathbf 0\}$ or $|\beta_\sigma|\ge 2^{d-3}$.  If
$\beta_\sigma=\{\mathbf 0\}$, then for every $S\in\Ucal_{n,d}$,
\[
    \tau_S=\sigma(x^S)=x^S(\mathbf 0)=0,
\]
which is exactly the claimed vanishing of every coordinate $y_S$.

It remains to rule out the case $|\beta_\sigma|\ge 2^{d-3}$.  For each $i$, let
$\beta_i:=\beta_{\sigma^{(i)}}$.  The multiplicativity constraint
$y_Sy_T=y_{S\cup T}$, satisfied in superposition, gives
\begin{equation}
\label{eq:monomial-superposition-multiplicativity-v3}
    \sigma(x^{S\cup T})=\tau_{S\cup T}
    =\sum_{i=1}^t \tau^{(i)}_S\tau^{(i)}_T
    =\sum_{i=1}^t \sigma^{(i)}(x^S)\sigma^{(i)}(x^T)
\end{equation}
whenever $S,T\in\Ucal_{n,d}$ and $|S\cup T|\le d$.  Using the same
constant-preserving extension, \eqref{eq:monomial-superposition-multiplicativity-v3}
implies, by checking monomial pairs and extending $\F_2$-bilinearly, that
\begin{equation}
\label{eq:polynomial-superposition-multiplicativity-v3}
    \sigma(gh)=\sum_{i=1}^t \sigma^{(i)}(g)\sigma^{(i)}(h)
\end{equation}
for every polynomial $g$ of degree at most $3d/4$ and every polynomial $h$ of
degree at most $d/4$ with zero constant term.  Indeed, since $h$ has zero constant term, it suffices to check $h=x^T$ and $g=1$ or $g=x^S$.  The case $g=x^S$ is exactly \eqref{eq:monomial-superposition-multiplicativity-v3}; the case $g=1$ is the coordinate identity $\tau_T=\sum_i\tau^{(i)}_T$.  The degree bounds ensure $|S\cup T|\le d$ in the first case.

To analyze \eqref{eq:polynomial-superposition-multiplicativity-v3} using Fourier analysis, we need to switch from $\F_2$ to the real values in $\{-1, 1\}$. 
For bits $u,v\in\F_2$, if $a=(-1)^u$ and $b=(-1)^v$, then
\[
    (-1)^{uv}=a\wedge b:=\frac{1+a+b-ab}{2}.
\]
Taking signs in \eqref{eq:polynomial-superposition-multiplicativity-v3} and
using \Cref{lem:monomial-assignments-as-points-v3}, we obtain, for every such
pair $g,h$,
\begin{equation}
\label{eq:ks-sign-identity-v3}
    \chi_{\beta_\sigma}(gh)
    \prod_{i=1}^t\left(\chi_{\beta_i}(g)\wedge \chi_{\beta_i}(h)\right)=1.
\end{equation}
Now choose $g$ and $h$ uniformly at random from the two polynomial spaces above.
The expectation of the left-hand side of \eqref{eq:ks-sign-identity-v3} is $1$.
On the other hand, expanding each factor
\[
    a\wedge b=\frac{1+a+b-ab}{2}
\]
expresses the same expectation as a sum of at most $4^t$ terms, each with
coefficient of absolute value at most $2^{-t}$ and each of the form
\[
    \E_{g,h}\bigl[\chi_{\beta_\sigma}(gh)\chi_\gamma(g)\chi_\alpha(h)\bigr]
\]
for some subsets $\gamma,\alpha\subseteq\F_2^{n+1}$.  By
\Cref{lem:ks-correlation-bound-v3} above, the absolute value of the whole expectation
is at most
\[
    2^t\cdot 2^{-2^{d/4-2}+1}.
\]
Choosing the absolute constant $c$ sufficiently large makes this quantity
strictly smaller than $1$ for every $t\ge 1$, contradicting
\eqref{eq:ks-sign-identity-v3}.  Hence the case
$|\beta_\sigma|\ge 2^{d-3}$ cannot occur, and the aggregate assignment must
vanish on every coordinate $y_S$.
\end{proof}

\subsection{Linearizing the superposition instance}
\label{subsec:linearized-superposition-instance}

Fix a degree parameter $d$, and let $\Qcal$ be the quadratic system from
\Cref{thm:homogeneous-superposition-v3}.  Its variables are
$\{y_S:S\in\Ucal_{n,d}\}$.  We define a linear subspace of matrices with rows and columns indexed by
$\Ucal_{n,d}$ by imposing two kinds of linear constraints.
\begin{enumerate}[label=(\arabic*)]
    \item \emph{Equal-union constraints:}
    \begin{equation}
    \label{eq:equal-union-constraints-v3}
        A_{S,T}=A_{S',T'}
        \qquad\text{whenever }S\cup T=S'\cup T'.
    \end{equation}
    \item \emph{Equation constraints:} for every equation of $\Qcal$ in the form
\begin{equation}
\label{eq:generic-qcal-equation-v3}
    \sum_{S,T\in\Ucal_{n,d}} c_{S,T}y_Sy_T
    +\sum_{R\in\Ucal_{n,d}} b_Ry_R=0,
\end{equation}
where $c_{S,T},b_R\in\F_2$, impose
    \begin{equation}
    \label{eq:qcal-moment-equation-constraints-v3}
        \sum_{S,T\in\Ucal_{n,d}} c_{S,T}A_{S,T}
        +\sum_{R\in\Ucal_{n,d}} b_RA_{R,R}=0.
    \end{equation}
\end{enumerate}
Let $\Lcal_d(\Qcal)$ be the set of matrices satisfying
\eqref{eq:equal-union-constraints-v3} and
\eqref{eq:qcal-moment-equation-constraints-v3}. This is a linear subspace of
$\F_2^{N\times N}$, where
\[
    N=|\Ucal_{n,d}|=\sum_{j=1}^d\binom{n+1}{j}.
\]
Furthermore, $\Lcal_d(\Qcal)$ only contains symmetric matrices due to the constraints $A_{S,T}=A_{T,S}$ in \eqref{eq:equal-union-constraints-v3}. 

\begin{restatable}[Decomposition into symmetric rank-one matrices,~\cite{KhotSaket2014}]{lemma}{symmetricdecomposition}

\label{lem:symmetric-decomposition-v3}
Let $A\in \F_2^{N\times N}$ be symmetric of rank $k$.  Then there is an integer
$t$ with $0\le t\le \lfloor 3k/2\rfloor$ and vectors
$u^{(1)},\dots,u^{(t)}\in\F_2^N$ such that
\[
    A=\sum_{i=1}^{t} u^{(i)}u^{(i)\top}.
\]
\end{restatable}
We include a proof of \Cref{lem:symmetric-decomposition-v3} in \Cref{sec:proof-khot-saket}.

\begin{lemma}[Low rank gives superposition assignments]
\label{lem:low-rank-superposition-v3}
Let $k\ge1$.  Let $A\in\Lcal_d(\Qcal)$ have rank at most $k$, and let
$t\ge\lfloor 3k/2\rfloor$.  Then there exist
$u^{(1)},\dots,u^{(t)}\in\F_2^{N}$ such that
\[
    A=\sum_{i=1}^t u^{(i)}u^{(i)\top},
\]
and the assignments $y_S\mapsto u^{(i)}_S$ satisfy every equation of $\Qcal$ in
superposition.
\end{lemma}

\begin{proof}
Since $A\in\Lcal_d(\Qcal)$ satisfies the equal-union constraints, it is
symmetric.  By \Cref{lem:symmetric-decomposition-v3}, $A$ is a sum of at most
$\lfloor 3k/2\rfloor\le t$ symmetric rank-one matrices.  Padding with zero vectors
if necessary, write
\[
    A=\sum_{i=1}^{t}u^{(i)}u^{(i)\top}.
\]

Fix an equation of $\Qcal$ written as in
\eqref{eq:generic-qcal-equation-v3}.  Since $A$ satisfies the corresponding
equation constraint \eqref{eq:qcal-moment-equation-constraints-v3}, we have
\[
    0=
    \sum_{S,T\in\Ucal_{n,d}} c_{S,T}A_{S,T}
    +\sum_{R\in\Ucal_{n,d}} b_RA_{R,R}.
\]
Substituting the rank-one decomposition and using $(u^{(i)}_R)^2=u^{(i)}_R$
over $\F_2$ gives
\[
\begin{aligned}
    0
    &=\sum_{i=1}^t
      \left(
      \sum_{S,T\in\Ucal_{n,d}} c_{S,T}u^{(i)}_Su^{(i)}_T
      +\sum_{R\in\Ucal_{n,d}} b_R(u^{(i)}_R)^2
      \right) \\
    &=\sum_{i=1}^t
      \left(
      \sum_{S,T\in\Ucal_{n,d}} c_{S,T}u^{(i)}_Su^{(i)}_T
      +\sum_{R\in\Ucal_{n,d}} b_Ru^{(i)}_R
      \right).
\end{aligned}
\]
The last expression is the superposition sum of this equation evaluated on the
assignments $y_S\mapsto u^{(i)}_S$.  Hence the assignments satisfy every equation
of $\Qcal$ in superposition.
\end{proof}

\begin{lemma}[Zero aggregate assignment forces vanishing]
\label{lem:zero-resultant-rank-v3}
Let $A \in\Lcal_d(\Qcal)$, $A \neq 0$, admit a decomposition
\[
    A=\sum_{i=1}^t u^{(i)}u^{(i)\top}
\]
such that its aggregate assignment is zero:
\[
    \sum_{i=1}^t u^{(i)}_U=0
    \qquad\text{for every }U\in\Ucal_{n,d}.
\]
Then 
\[
    \operatorname{rank}(A) \ge \binom{d+1}{\floorbra{(d+1)/2}} \ . 
\]
\end{lemma}

\begin{proof}
The zero-aggregation condition gives
\[
    A_{U,U}=\sum_{i=1}^t (u^{(i)}_U)^2
           =\sum_{i=1}^t u^{(i)}_U=0,
    \qquad \forall U\in\Ucal_{n,d}.
\]
If $U,V\in\Ucal_{n,d}$ and $|U\cup V|\le d$, then $U\cup V\in\Ucal_{n,d}$, and
the equal-union constraints give
\begin{equation}
\label{eq:small-union-zero-v3}
    A_{U,V}=A_{U\cup V,U\cup V}=0.
\end{equation}

Suppose $A\ne0$, and choose $U,V$ with $A_{U,V}=1$ minimizing
$s:=|U\cup V|$.  By \eqref{eq:small-union-zero-v3}, we have $d<s\le2d$.  Put
$R:=U\cup V$.  For every pair $U',V'\in\Ucal_{n,d}$ with $U'\cup V'=R$, the
equal-union constraints give $A_{U',V'}=A_{U,V}=1$.

Let
\[
    \Fcal:=\{F\subseteq R:|F|=\floorbra{s/2}\},
    \qquad
    \mathcal{G}:=\{G\subseteq R:|G|=\ceilbra{s/2}\}.
\]
Because $d<s\le2d$, every set in $\Fcal\cup\mathcal{G}$ lies in
$\Ucal_{n,d}$.  In the submatrix $A|_{\Fcal,\mathcal{G}}$, the entry indexed by
$(F,G)$ is $1$ exactly when $G=R\setminus F$: in that case $F\cup G=R$, while
all other pairs have union a proper subset of $R$ and hence have entry $0$ by
the minimality of $s$.  Thus $A|_{\Fcal,\mathcal{G}}$ is a permutation matrix,
so
\[
    \operatorname{rank}(A)
    \ge \binom{s}{\floorbra{s/2}}
    \ge \binom{d+1}{\floorbra{(d+1)/2}} \ . \qedhere
\]
\end{proof}

\begin{theorem}[Rank gap over $\F_2$]
\label{thm:rank-gap-subspace-f2-v3}
For every integer $k\ge1$, there is a deterministic
$n^{O(\log k)}$-time reduction from a \textup{\textsc{3Sat}} instance on
$n$ variables to a linear subspace
\[
    \Lcal\subseteq\F_2^{N\times N},
    \qquad N=n^{O(\log k)},
\]
such that satisfiable instances yield a nonzero rank-$1$ matrix in $\Lcal$, and
unsatisfiable instances yield no nonzero matrix in $\Lcal$ of rank at most $k$.
\end{theorem}

\begin{proof}
Set
\[
    t:=\left\lfloor\frac{3k}{2}\right\rfloor,
\]
and choose $d$ to be a positive multiple of $4$ such that 
\[
    d\ge \max\{8,c\log(t+1)\} 
    \qquad\text{and}\qquad
    \binom{d+1}{\floorbra{(d+1)/2}}>k,
\]
where $c$ is the constant from \Cref{thm:homogeneous-superposition-v3}.  Then
$d=O(\log k)$.

Apply \Cref{thm:homogeneous-superposition-v3} with this $d$ and output
$\Lcal:=\Lcal_d(\Qcal)$.  The construction time is
$n^{O(d)}=n^{O(\log k)}$, and
\[
    N=|\Ucal_{n,d}|=\sum_{j=1}^d\binom{n+1}{j}=n^{O(d)}=n^{O(\log k)}.
\]

In the YES case, take a satisfying Boolean assignment to the original formula,
set $a_0=1$, and define $y_S:=a^S$ for every $S\in\Ucal_{n,2d}$.  Then
$H_d(y)=v_d(a)v_d(a)^\top$ is nonzero of rank $1$.  Since every equation of
$\Qcal$ vanishes at $(a^S)_{S\in\Ucal_{n,d}}$, the matrix $H_d(y)$ satisfies both
the equal-union constraints and the equation constraints.  Hence
$H_d(y)\in\Lcal$.

In the NO case, suppose $A\in\Lcal$ is nonzero and $\operatorname{rank}(A) \le k$.  By \Cref{lem:low-rank-superposition-v3}, there
are exactly $t$ assignments whose rank-one sum is $A$ and which satisfy all
equations of $\Qcal$ in superposition.  By
\Cref{thm:homogeneous-superposition-v3}, their aggregate assignment is zero.
Then \Cref{lem:zero-resultant-rank-v3} and the choice of $d$ imply that $
    \operatorname{rank}(A) > k$. 
\end{proof}

\subsection{Extending the hardness from $\mathbb F_2$ to $\mathbb F_{2^r}$}
\label{subsec:f2tof2r}

The following descent lemma is the black-box field-extension step.  It loses a
factor of $r$ in the rank parameter, which is why the base-field gap below is
run with parameter $rk$.

\begin{lemma}[Rank descent]
\label{lem:rank-descent-extension-v3}
Let $\mathbb K=\F_{2^r}$, and let $\Lcal\subseteq\F_2^{N\times N}$ be a linear
subspace.  If $A\in \Lcal\otimes_{\F_2}\mathbb K$ is nonzero and
$\operatorname{rank}_{\mathbb K}(A)\le k$, then there exists a nonzero
$B\in\Lcal$ with
\[
    \operatorname{rank}_{\F_2}(B)\le rk.
\]
\end{lemma}

\begin{proof}
Choose an $\F_2$-linear map $\phi:\mathbb K\to\F_2$ such that the entrywise
matrix $\phi(A)$ is nonzero; this is possible because $A$ has a nonzero entry.
Set $B:=\phi(A)$.  If
$A=\sum_j \alpha_j A_j$ with $\alpha_j\in\mathbb K$ and $A_j\in\Lcal$, then
\[
    B=\sum_j \phi(\alpha_j)A_j,
\]
so $B\in\Lcal$.

It remains to bound the rank of $B$.  Write
$A=\sum_{j=1}^k p_jq_j^\top$ with $p_j,q_j\in\mathbb K^N$.  Fix an
$\F_2$-basis $\theta_1,\ldots,\theta_r$ of $\mathbb K$.  For each $j$, let
$P_j,Q_j\in\F_2^{N\times r}$ be the coordinate matrices of $p_j$ and $q_j$ in
this basis, and let $M_\phi\in\F_2^{r\times r}$ be given by
\[
    (M_\phi)_{ab}:=\phi(\theta_a\theta_b).
\]
Then
\[
    \phi(p_jq_j^\top)=P_jM_\phi Q_j^\top,
\]
which has rank at most $r$ over $\F_2$.  Therefore
\[
    \operatorname{rank}_{\F_2}(B)
    \le \sum_{j=1}^k \operatorname{rank}_{\F_2}\bigl(\phi(p_jq_j^\top)\bigr)
    \le rk . \qedhere
\]
\end{proof}

\begin{proof}[Proof of \Cref{thm:rank-gap-subspace-v3}]
Let $\mathbb K=\F_{2^r}$ and set $k_0:=rk$.  Apply
\Cref{thm:rank-gap-subspace-f2-v3} with rank parameter $k_0$ to obtain
$\Mcal\subseteq\F_2^{N\times N}$, and output
\[
    \Lcal:=\Mcal\otimes_{\F_2}\mathbb K\subseteq\mathbb K^{N\times N},
\]
equivalently the same homogeneous linear equations interpreted over
$\mathbb K$.

Completeness is preserved under field extension: a nonzero rank-$1$ matrix in
$\Mcal$ remains a nonzero rank-$1$ matrix in $\Lcal$.  For soundness, suppose
that the input formula is unsatisfiable and that $A\in\Lcal$ is nonzero with
$\operatorname{rank}_{\mathbb K}(A)\le k$.  By
\Cref{lem:rank-descent-extension-v3}, there is a nonzero
$B\in\Mcal$ with $\operatorname{rank}_{\F_2}(B)\le rk=k_0$, contradicting the
soundness of $\Mcal$.

The running time and dimension are
\[
    n^{O(\log k_0)}=n^{O(\log(rk))},
    \qquad
    N=n^{O(\log k_0)}=n^{O(\log(rk))}.
\]
For fixed extension degree $r$, this is $n^{O(\log k)}$ time and
$N=n^{O(\log k)}$.
\end{proof}

The inapproximability statement \Cref{thm:intro-main-inapproximability} is a canonical corollary of \Cref{thm:rank-gap-subspace-v3}; we defer the proof to \Cref{app:inapproximability-corollaries}.

\begin{remark}[Toward arbitrary finite fields]
\label{rem:toward-arbitrary-fq-strong}
It is natural to ask whether the $n^{O(\log k)}$ rank-gap reduction of this section can be obtained over an arbitrary fixed finite field $\F_q$. We plan to include this extension in a later version of this paper.
The main technical step to achieve this is replacing the binary low-degree long code used in the Dinur--Guruswami and Khot--Saket frameworks \cite{DG15,KhotSaket2014} with a suitable $q$-ary analogue.
In particular, carrying out the superposition-soundness argument over $\F_q$ requires extending the relevant low-error Reed--Muller testing and, crucially, generalizing the correlation bound of \Cref{lem:ks-correlation-bound-v3} to arbitrary finite fields.
\end{remark}

%% file: sec/moment.tex
\section{An $n^{O(k)}$ Construction over any $\mathbb F_q$}
\label{sec:moment-matrix-construction}

In this section, we start with the canonical NP-hardness of Boolean \textup{\textsc{QuadEq}}, and prove \Cref{thm:elementary-moment-rank-gap} by a direct moment-matrix construction. 

\begin{lemma}[Boolean \textup{\textsc{QuadEq}} hardness, \cite{FY79,GJ79}]
\label{lem:quadeq-source-elementary}
Let $\F_q$ be any finite field.  Given squarefree quadratic polynomials
$f_1,\ldots,f_m$ over $\F_q$ in variables $x_1,\ldots,x_n$, it is NP-hard to
decide whether there exists
$a\in\{0,1\}^{n}\subseteq\F_q^{n}$ such that
$f_1(a)=\cdots=f_m(a)=0$.
Equivalently, each input polynomial may be written as
\[
    f_\ell(x)=\sum_{U\in\Vcal_{n,2}} c_{\ell,U}x^U,
    \qquad c_{\ell,U}\in\F_q,
\]
where the coordinate indexed by $\emptyset$ is the constant term.
\end{lemma}

\begin{theorem}
\label{thm:elementary-moment-rank-gap}
Fix a finite field $\F_q$. For every integer $k\ge 1$, there is a deterministic
$n^{O(k)}$-time reduction from a Boolean
\textup{\textsc{QuadEq}} instance with $n$ variables over
$\F_q$ to a linear subspace
\[
    \Lcal\subseteq\F_q^{N\times N},
    \qquad
    N=n^{O(k)},
\]
such that:
\begin{itemize}
    \item (YES) if the Boolean \textup{\textsc{QuadEq}} instance is satisfiable, then
    $\Lcal$ contains a nonzero matrix of rank $1$;
    \item (NO) if the Boolean \textup{\textsc{QuadEq}} instance is unsatisfiable, then
    $\Lcal$ contains no nonzero matrix of rank at most $k$.
\end{itemize}
\end{theorem}

Throughout this section all vector spaces, matrices, and polynomials are over a
fixed finite field $\F_q$.  We use the Boolean monomial and pseudo-moment
conventions from \Cref{sec:preliminaries}; in particular, $\Vcal_{n,d}$ includes
the empty set.

\subsection{The pseudo-moment subspace}
\label{subsec:pseudo-moment-subspace}

Fix an integer $d\ge 1$.  Following the notation of
\Cref{sec:preliminaries}, a degree-$2d$ pseudo-moment vector is
$y=(y_S)_{S\in\Vcal_{n,2d}}$, and its associated pseudo-moment matrix is
$H_d(y)$.

We define a linear subspace of matrices with rows and columns indexed by
$\Vcal_{n,d}$ by imposing two kinds of linear constraints.
\begin{enumerate}[label=(\arabic*)]
    \item \emph{Equal-union constraints:}
    \begin{equation}
    \label{eq:equal-union-constraints-elementary}
        A_{S,T}=A_{S',T'}
        \qquad\text{whenever }S\cup T=S'\cup T'.
    \end{equation}
    After imposing these constraints, for every $R\in\Vcal_{n,2d}$ we write
    $y_R$ for the common value $A_{S,T}$ over all pairs
    $S,T\in\Vcal_{n,d}$ with $S\cup T=R$. For convenience, we describe the next set of constraints in terms of $y$.

    \item \emph{Localizing constraints:}
    For each \textsc{QuadEq} equation
    \[
        f_\ell(x)=\sum_{U\in\Vcal_{n,2}} c_{\ell,U}x^U,
        \qquad c_{\ell,U}\in\F_q,
    \]
    we impose the linear constraints
    \begin{equation}
    \label{eq:moment-localizing-elementary}
        \sum_{U\in\Vcal_{n,2}} c_{\ell,U}\,y_{U\cup W}=0
        \qquad
        \text{for every }\ell\in[m]\text{ and every }W\in\Vcal_{n,2d-2}.
    \end{equation}
    These are the moment
    analogues of the equations $x^W f_\ell(x)=0$: after multiplying $f_\ell$ by
    any squarefree monomial of degree at most $2d-2$, the corresponding linear
    combination of pseudo-moments must vanish.
\end{enumerate}
Let
$\Lcal_d(f_1,\ldots,f_m)$ denote the set of matrices satisfying the constraints above.  This is a linear subspace of
$\F_q^{N \times N}$, where
\[
    N=|\Vcal_{n,d}|=\sum_{j=0}^{d}\binom{n}{j}=n^{O(d)}.
\]

For completeness, if $a=(a_1,\ldots,a_n)\in\{0,1\}^{n}$ satisfies all equations
$f_\ell(a)=0$, set $y_S:=a^S$ for every $S\in\Vcal_{n,2d}$.  Then
$H_d(y)=v_d(a)v_d(a)^\top$ has rank $1$ and is nonzero because the coordinate
indexed by $\emptyset$ is $1$.  Moreover, for every $W\in\Vcal_{n,2d-2}$,
\[
    \sum_{U\in\Vcal_{n,2}} c_{\ell,U}y_{U\cup W}
    =a^W f_\ell(a)=0,
\]
so $H_d(y)\in\Lcal_d(f_1,\ldots,f_m)$.

\smallskip
We now begin the soundness proof.  It suffices to show that every nonzero matrix in $\Lcal_d(f_1,\ldots,f_m)$ with rank at most $d$ yields an actual Boolean solution of the source
instance.  Let $y$ be the unique pseudo-moment vector associated with such a
matrix $A$, so $A=H_d(y)$ and $y$ satisfies
\eqref{eq:moment-localizing-elementary}.  
For a genuine Boolean point $a$, multiplying the label of a monomial
column by $x_i$ multiplies that column by the scalar $a_i$; equivalently, the
column space is invariant under multiplication by each variable.  We will find two consecutive 
levels where low rank forces a similar form of invariance. 

For $0\le e\le d$, let
\[
    H_e(y):=\bigl(y_{S\cup T}\bigr)_{S,T\in\Vcal_{n,e}},
    \qquad
    r_e:=\operatorname{rank} H_e(y),
\]
and let $C_e$ denote the column space of $H_e(y)$.

The soundness argument has four steps.
\begin{enumerate}[label=\textup{(\arabic*)},leftmargin=2em]
    \item Low rank forces a nonzero \emph{flat level}, meaning a level $e$ with
    $r_e=r_{e+1}>0$.
    \item At such a flat level, multiplying a basis column by any variable
    $x_i$ creates no new column direction.  This defines a linear map
    $T_i:C_e\to C_e$, interpreted as multiplication by $x_i$.
    \item Because the relevant column relations are certified one level up (a.k.a inside $C_{e+1}$), these maps satisfy the Boolean multiplication rules
    $T_i^2=T_i$ and $T_iT_j=T_jT_i$, and the localizing equations become the
    operator identities $f_\ell(T)=0$. 
    \item The commuting projections $T_i$ have a common eigenvector.  Their
    eigenvalues form a Boolean point $a\in\{0,1\}^{n}\subseteq\F_q^{n}$, and
    the identities $f_\ell(T)=0$ force $f_\ell(a)=0$.
\end{enumerate}

\subsection{Finding a nonzero flat level}
\label{subsec:nonzero-flat-level}

Note that the ranks $r_0,r_1,\ldots,r_d$ are non-decreasing, and the final
rank is at most $d$.  Monotonicity alone would still allow the ranks to be zero
for many initial levels and then grow by one at each remaining level.  The key
point is that a nonzero pseudo-moment cannot start growing this slowly.  If
$S$ is a minimum-size set with $y_S\ne0$, then the same permutation-submatrix
idea as in \Cref{lem:zero-resultant-rank-v3} already forces, at level
$s=\ceilbra{|S|/2}$, the lower bound $r_s\ge s+1$ (indeed
$r_s\ge 2^{\Omega(s)}$ except for the smallest cases).  If every subsequent
positive step were strict, then
\[
    r_d\ge r_s+(d-s)\ge (s+1)+(d-s)=d+1,
\]
contradicting $r_d\le d$.  Thus some nonzero plateau $r_e=r_{e+1}$ must occur.

\begin{lemma}
\label{lem:elementary-flat-level}
Let $d\ge 1$.  Suppose $H_d(y)$ is nonzero and
$\operatorname{rank}H_d(y)\le d$.  Then there exists
$e\in\{0,1,\ldots,d-1\}$ such that
\[
    r_e=r_{e+1}>0.
\]
\end{lemma}

\begin{proof}
Choose a set $S\in\Vcal_{n,2d}$ of minimum size such that $y_S\ne 0$, and put
\[
    \lambda:=y_S,
    \qquad h:=|S|,
    \qquad s:=\ceilbra{h/2}.
\]
Since $h\le 2d$, we have $s\le d$.  The first step is the following rank lower
bound forced by the first nonzero moment:
\begin{equation}
\label{eq:rs-large-from-first-nonzero-moment}
    r_s\ge s+1.
\end{equation}

If $h=0$, then $s=0$ and $y_\emptyset=\lambda\ne 0$, so
$r_s=r_0=1=s+1$.  If $h=1$, say $S=\{i\}$, then the minimality of $S$ gives
$y_\emptyset=0$ and $y_S=\lambda$.  The submatrix of $H_1(y)$ with rows and
columns indexed by $\emptyset,S$ is
\[
    \begin{pmatrix}
        y_\emptyset & y_S\\
        y_S & y_S
    \end{pmatrix}
    =
    \begin{pmatrix}
        0&\lambda\\
        \lambda&\lambda
    \end{pmatrix},
\]
which has rank $2=s+1$.

It remains to consider $h\ge 2$.  Use the subsets $A\subseteq S$ of size
$\floorbra{h/2}$ as row labels and the subsets $B\subseteq S$ of size
$\ceilbra{h/2}=s$ as column labels.  The resulting submatrix of $H_s(y)$ has
entries $y_{A\cup B}$.  By the minimality of $S$, this entry is $\lambda$ exactly
when $A\cup B=S$, and is $0$ otherwise.  For each row $A$, there is a unique such
column, namely $B=S\setminus A$.  Thus this submatrix is $\lambda$ times a
permutation matrix of size
\[
    \binom{h}{\floorbra{h/2}}.
\]
Therefore
\[
    r_s\ge \binom{h}{\floorbra{h/2}}
        \ge h
        \ge \ceilbra{h/2}+1
        =s+1,
\]
where the middle inequality uses $h\ge 2$.  This proves
\eqref{eq:rs-large-from-first-nonzero-moment} in all cases.

The matrices $H_0(y),H_1(y),\ldots,H_d(y)$ are nested principal submatrices, so
the ranks $r_0,r_1,\ldots,r_d$ are nondecreasing.  Suppose, for contradiction,
that no $e\in\{0,1,\ldots,d-1\}$ satisfies $r_e=r_{e+1}>0$.  Since $r_s>0$ and
the ranks are nondecreasing, the ranks strictly
increase at every such step.  Thus
\[
    r_d\ge r_s+(d-s)\ge (s+1)+(d-s)=d+1,
\]
contradicting $r_d=\operatorname{rank}H_d(y)\le d$.  Therefore a nonzero flat
level exists.
\end{proof}

\subsection{Multiplication operators and localizing identities}
\label{subsec:multiplication-operators-localizing-identities}

For $0\le e\le d$ and a set $A\subseteq[n]$ with $|A|\le 2d-e$, define the
truncated column
\[
    \mathbf c_e(A):=\bigl(y_{R\cup A}\bigr)_{R\in\Vcal_{n,e}}
    \in \F_q^{\Vcal_{n,e}}.
\]
If $A\in\Vcal_{n,e}$, then $\mathbf c_e(A)$ is exactly the column of $H_e(y)$
indexed by $A$.  If $|A|=t>e$, it is not literally a column of $H_e(y)$, but it is still the degree-$e$ truncation of the column indexed by $A$, in a larger moment matrix $H_t(y)$.

We now explain how the flat level is used. 
At a flat level $r_e=r_{e+1}$, any basis of the column space $C_e$ remains a basis after lifting the columns to level $e+1$.  Hence multiplying the column label by $x_i$ gives a well-defined operator $T_i:C_e\to C_e$. Furthermore, this extra-level lift certifies that
\[
    T_i\mathbf c_e(A)=\mathbf c_e(A\cup\{i\})
\]
for every $A\in\Vcal_{n,e+1}$. 
This allows us to multiply the column label by a second variable consistently, yielding $T_i^2=T_i$, $T_iT_j=T_jT_i$, and ultimately translating the localizing constraints into the operator identities $f_\ell(T)=0$. 

\begin{lemma}
\label{lem:elementary-flatness-lifts-multiplication}
Assume $r_e=r_{e+1}$ and $e\le d-1$.  Let
$\Ical\subseteq\Vcal_{n,e}$ be such that the vectors
$\{\mathbf c_e(B):B\in\Ical\}$ form a basis for $C_e$.  Then, for every
$i\in[n]$, there is a linear map $T_i:C_e\to C_e$ satisfying
\[
    T_i\mathbf c_e(B)=\mathbf c_e(B\cup\{i\})
    \qquad\text{for every }B\in\Ical.
\]
Moreover, for every $A\in\Vcal_{n,e+1}$ and every
$i\in[n]$, the vector $\mathbf c_e(A)$ lies in $C_e$ and
\begin{equation}
\label{eq:Ti-all-truncated-columns}
    T_i\mathbf c_e(A)=\mathbf c_e(A\cup\{i\}).
\end{equation}
\end{lemma}

\begin{proof}
First observe that the longer vectors
$\{\mathbf c_{e+1}(B):B\in\Ical\}$ are linearly independent: any linear relation
among them restricts, on the rows indexed by $\Vcal_{n,e}$, to the same linear
relation among the basis vectors $\{\mathbf c_e(B):B\in\Ical\}$.  Since
$r_e=r_{e+1}$, these longer vectors form a basis for the column space of
$H_{e+1}(y)$.

For every $i\in[n]$ and every $B\in\Ical$, the vector
$\mathbf c_{e+1}(B\cup\{i\})$ is a column of $H_{e+1}(y)$, and hence has a unique
expansion in the basis $\{\mathbf c_{e+1}(P):P\in\Ical\}$.  Restricting that
expansion to rows of degree at most $e$ shows that
$\mathbf c_e(B\cup\{i\})\in C_e$.
We may therefore define $T_i$ on the basis by
\[
    T_i\mathbf c_e(B):=\mathbf c_e(B\cup\{i\})
    \qquad(B\in\Ical),
\]
and extend linearly to all of $C_e$.

Now fix $A\in\Vcal_{n,e+1}$.  We would like to argue that
$T_i\mathbf c_e(A)=\mathbf c_e(A\cup\{i\})$.
Since $\mathbf c_{e+1}(A)$ is a column of $H_{e+1}(y)$, there are unique
coefficients $\lambda_{A,B}\in\F_q$ such that
\[
    \mathbf c_{e+1}(A)
    =\sum_{B\in\Ical}\lambda_{A,B}\mathbf c_{e+1}(B).
\]
Restricting to rows in $\Vcal_{n,e}$ gives
\[
    \mathbf c_e(A)
    =\sum_{B\in\Ical}\lambda_{A,B}\mathbf c_e(B),
\]
so $\mathbf c_e(A)\in C_e$.  By linearity and the definition of $T_i$ on the
basis,
\[
    T_i\mathbf c_e(A)
    =\sum_{B\in\Ical}\lambda_{A,B}T_i\mathbf c_e(B)
    =\sum_{B\in\Ical}\lambda_{A,B}\mathbf c_e(B\cup\{i\}).
\]
It remains to identify the last expression with $\mathbf c_e(A\cup\{i\})$.  Fix
a row label $R\in\Vcal_{n,e}$.  Since $R\cup\{i\}\in\Vcal_{n,e+1}$, evaluating
the expansion of $\mathbf c_{e+1}(\cdot)$ at the row $R\cup\{i\}$ gives
\[
\begin{aligned}
    \sum_{B\in\Ical}\lambda_{A,B}\bigl(\mathbf c_e(B\cup\{i\})\bigr)_R
    &=\sum_{B\in\Ical}\lambda_{A,B}y_{R\cup B\cup\{i\}} \\
    &=\sum_{B\in\Ical}\lambda_{A,B}\bigl(\mathbf c_{e+1}(B)\bigr)_{R\cup\{i\}} \\
    &=\bigl(\mathbf c_{e+1}(A)\bigr)_{R\cup\{i\}} \\
    &=y_{R\cup A\cup\{i\}}
     =\bigl(\mathbf c_e(A\cup\{i\})\bigr)_R .
\end{aligned}
\]
This holds for every row $R\in\Vcal_{n,e}$, so
$T_i\mathbf c_e(A)=\mathbf c_e(A\cup\{i\})$.  The vector on the right is
well-defined because $|A\cup\{i\}|\le e+2\le 2d-e$, using $e\le d-1$.
\end{proof}

\begin{lemma}
\label{lem:elementary-multiplication-flat-level}
Assume $r_e=r_{e+1}>0$ and $e\le d-1$.  Let $C_e$ be the column space of
$H_e(y)$.  If $y$ satisfies the equations
\eqref{eq:moment-localizing-elementary}, then there are linear maps
\[
    T_1,\ldots,T_n:C_e\to C_e
\]
with the following properties:
\begin{enumerate}[label=(\roman*)]
    \item $T_i^2=T_i$ for every $i$;
    \item $T_iT_j=T_jT_i$ for every $i,j$;
    \item for every $\ell\in[m]$,
    \[
        \sum_{U\in\Vcal_{n,2}} c_{\ell,U}T_U=0
        \qquad\text{as a linear map on }C_e,
    \]
    where $T_U:=\prod_{i\in U}T_i$ and $T_\emptyset$ is the identity map.
\end{enumerate}
\end{lemma}

\begin{proof}
Choose $\Ical\subseteq\Vcal_{n,e}$ such that
$\{\mathbf c_e(B):B\in\Ical\}$ is a basis for $C_e$, and define the maps
$T_i$ as in \Cref{lem:elementary-flatness-lifts-multiplication}.  By
\Cref{lem:elementary-flatness-lifts-multiplication}, for every
$A\in\Vcal_{n,e+1}$ and every $i$,
\begin{equation}
\label{eq:Ti-multiplies-all-visible-columns}
    T_i\mathbf c_e(A)=\mathbf c_e(A\cup\{i\}).
\end{equation}

We first prove the Boolean multiplication rules.  Fix $B\in\Ical$ and
$i,j\in[n]$.  Taking $A=B$ in
\eqref{eq:Ti-multiplies-all-visible-columns} gives
$T_i\mathbf c_e(B)=\mathbf c_e(B\cup\{i\})$.  Since
$B\cup\{i\}\in\Vcal_{n,e+1}$, applying
\eqref{eq:Ti-multiplies-all-visible-columns} with $A=B\cup\{i\}$ and variable
$j$ gives
\begin{equation}
\label{eq:TjTi-by-lifting-lemma}
    T_jT_i\mathbf c_e(B)
    =T_j\mathbf c_e(B\cup\{i\})
    =\mathbf c_e(B\cup\{i,j\}).
\end{equation}

Taking $j=i$ in \eqref{eq:TjTi-by-lifting-lemma} gives
$T_i^2\mathbf c_e(B)=T_i\mathbf c_e(B)$, because
$B\cup\{i,i\}=B\cup\{i\}$.  Since this holds on the basis vectors
$\{\mathbf c_e(B):B\in\Ical\}$, we get $T_i^2=T_i$.  Swapping $i$ and $j$ in
\eqref{eq:TjTi-by-lifting-lemma} gives
\[
    T_iT_j\mathbf c_e(B)=\mathbf c_e(B\cup\{j,i\})
    =\mathbf c_e(B\cup\{i,j\})=T_jT_i\mathbf c_e(B),
\]
again on every basis vector.  Hence $T_iT_j=T_jT_i$.

It remains to translate the original equations into operator identities.
Because the $T_i$ commute, $T_U:=\prod_{i\in U}T_i$ is well-defined.  From
\eqref{eq:Ti-multiplies-all-visible-columns} when $|U|\le 1$, and from
\eqref{eq:TjTi-by-lifting-lemma} when $|U|=2$, we have
\[
    T_U\mathbf c_e(B)=\mathbf c_e(B\cup U)
    \qquad\text{for every }B\in\Ical\text{ and every }U\in\Vcal_{n,2}.
\]
Therefore, for every row label $R\in\Vcal_{n,e}$,
\[
\begin{aligned}
    \left(\sum_{U\in\Vcal_{n,2}} c_{\ell,U}T_U\mathbf c_e(B)\right)_R
    &=\sum_{U\in\Vcal_{n,2}} c_{\ell,U}y_{R\cup B\cup U}.
\end{aligned}
\]
Here $R\cup B$ has size at most $2e$.  Since $e\le d-1$, we have
$2e\le 2d-2$, so the localizing equations
\eqref{eq:moment-localizing-elementary} apply with $W=R\cup B$.  Hence the
last sum is $0$.  Thus $\sum_{U\in\Vcal_{n,2}} c_{\ell,U}T_U$ kills every basis
vector of $C_e$, and therefore it is the zero map on $C_e$.
\end{proof}

\subsection{Decoding a satisfying assignment from commuting projections}
\label{subsec:decoding-commuting-projections}

The next lemma records the common-eigenvector property for commuting projections.

\begin{lemma}
\label{lem:elementary-common-eigenvector}
Let $C$ be a nonzero finite-dimensional vector space over $\F_q$, and let
$T_1,\ldots,T_n$ be commuting linear maps on $C$ with $T_i^2=T_i$ for every
$i$.  Then there are $v\in C\setminus\{0\}$ and scalars
$a_i\in\{0,1\}\subseteq\F_q$, $i\in[n]$, such that
\[
    T_i v=a_i v
    \qquad\text{for every }i=1,\ldots,n.
\]
\end{lemma}

\begin{proof}
We shrink the space one operator at a time.  Start with $C^{(0)}:=C$.  At step
$i=1,\ldots,n$, suppose that $C^{(i-1)}$ is nonzero and is preserved by
$T_i,T_{i+1},\ldots,T_n$.  Since $C^{(i-1)}$ is preserved by $T_i$, the map
$T_i$ restricts to a projection on $C^{(i-1)}$.  Hence
\[
    C^{(i-1)}=\bigl(C^{(i-1)}\cap\ker(T_i)\bigr)
    \oplus \bigl(C^{(i-1)}\cap\operatorname{im}(T_i)\bigr).
\]
Since $C^{(i-1)}$ is nonzero, at least one of the two summands above is
nonzero.  Choose a nonzero summand and call it $C^{(i)}$.  The new subspace is
still preserved by $T_{i+1},\ldots,T_n$, because those operators commute with
$T_i$.  If we chose the kernel summand, then every $w\in C^{(i)}$ satisfies
$T_iw=0$.  If we chose the image summand, then every $w\in C^{(i)}$ satisfies
$T_iw=w$, since $w=T_iz$ implies $T_iw=T_i^2z=T_iz=w$.

After doing this for $i=1,\ldots,n$, the final space $C^{(n)}$ is still
nonzero.  Pick any nonzero $v\in C^{(n)}$.  For each $i$, the construction
places $v$ either in $\ker(T_i)$ or in $\operatorname{im}(T_i)$, so $T_iv=0$ or
$T_iv=v$.  Thus $T_iv=a_iv$ for some $a_i\in\{0,1\}\subseteq\F_q$ for every
$i$.
\end{proof}

\begin{lemma}
\label{lem:elementary-decode-eigenvector}
Let $C$ be a nonzero finite-dimensional vector space over $\F_q$, and let
$T_1,\ldots,T_n$ be commuting linear maps on $C$.  Suppose that, for every
$\ell\in[m]$,
\[
    \sum_{U\in\Vcal_{n,2}} c_{\ell,U}T_U=0
    \qquad\text{as a linear map on }C,
\]
where $T_U:=\prod_{i\in U}T_i$ and $T_\emptyset$ is the identity map.  If there
are $v\in C\setminus\{0\}$ and scalars
$a_1,\ldots,a_n\in\{0,1\}\subseteq\F_q$ such that $T_iv=a_iv$ for every $i$,
then $a=(a_1,\ldots,a_n)$ satisfies
$f_1(a)=\cdots=f_m(a)=0$.
\end{lemma}

\begin{proof}
Since the operators commute, for every $U\in\Vcal_{n,2}$ we have
\[
    T_Uv=a^Uv.
\]
Therefore, for every $\ell$,
\[
    0=\Bigl(\sum_{U\in\Vcal_{n,2}}c_{\ell,U}T_U\Bigr)v
     =\Bigl(\sum_{U\in\Vcal_{n,2}}c_{\ell,U}a^U\Bigr)v
     =f_\ell(a)v.
\]
Since $v\ne 0$, the scalar $f_\ell(a)$ must be zero.  Hence the Boolean point
$a=(a_1,\ldots,a_n)$ satisfies all equations.
\end{proof}

We are now ready to prove \Cref{thm:elementary-moment-rank-gap}.

\begin{proof}[Proof of \Cref{thm:elementary-moment-rank-gap}]
Given the Boolean \textup{\textsc{QuadEq}} instance $f_1,\ldots,f_m$, set $d:=k$
and output
\[
    \Lcal:=\Lcal_d(f_1,\ldots,f_m).
\]
The matrix size is
\[
    N=|\Vcal_{n,d}|=\sum_{j=0}^{k}\binom{n}{j}=n^{O(k)},
\]
and the number of equal-union and localizing constraints is also
$(n+m)n^{O(k)}$, so the construction has the claimed running time.

If the Boolean \textup{\textsc{QuadEq}} instance has a solution, the
completeness argument above shows that $\Lcal$ contains a nonzero rank-$1$
matrix.  If the Boolean \textup{\textsc{QuadEq}} instance has no solution and
$A\in\Lcal$ is nonzero with $\operatorname{rank}(A)\le k$, then the unique
pseudo-moment vector $y$ associated with $A$ satisfies
$A=H_k(y)$ and \eqref{eq:moment-localizing-elementary}.  By
\Cref{lem:elementary-flat-level}, choose a flat level $e$ with
$r_e=r_{e+1}>0$.  By \Cref{lem:elementary-multiplication-flat-level}, the
column space $C_e$ carries commuting projections $T_1,\ldots,T_n$, and for
every $\ell\in[m]$,
\[
    \sum_{U\in\Vcal_{n,2}} c_{\ell,U}T_U=0.
\]
By \Cref{lem:elementary-common-eigenvector}, there are a nonzero vector
$v\in C_e$ and scalars $a_1,\ldots,a_n\in\{0,1\}\subseteq\F_q$ such that
$T_iv=a_iv$ for every $i$.  By \Cref{lem:elementary-decode-eigenvector},
$a$ is a Boolean solution, a contradiction.  Hence no such nonzero matrix
exists.
\end{proof}

The inapproximability statement \Cref{thm:intro-any-field-inapproximability} is a standard corollary of \Cref{thm:elementary-moment-rank-gap}; we defer the proof to Appendix~\ref{app:inapproximability-corollaries}.

%% file: sec/inapprox_corollaries.tex
\section{Proof of Inapproximability Corollaries}
\label{app:inapproximability-corollaries}

In this appendix we derive the two inapproximability statements from the two
rank-gap reductions proved in the body.

\maininapproximability*

\begin{proof}
Let $M$ denote the size of the input \textup{\textsc{3Sat}} instance.  By
\Cref{thm:rank-gap-subspace-v3}, there is an absolute constant $C$ such that the
reduction with rank parameter $k$ has matrix dimension and total output size at
most
\[
    N\le M^{C\log k}.
\]
In the YES case the optimum is $1$, while in the NO case it is larger than $k$.

For a constant factor $\gamma>1$, choose a constant integer $k\ge \gamma$.  The
reduction is polynomial time, and a polynomial-time distinguisher for the gap
$1$ versus $>\gamma$ would decide \textup{\textsc{3Sat}} in polynomial time.
This proves the first item.

Next fix $0<\epsilon<1$ and suppose that there is a polynomial-time distinguisher
for gap
\[
    \gamma(N)=2^{(\log N)^{1-\epsilon}}.
\]
Choose
\[
    k=\left\lceil 2^{(\log M)^a}\right\rceil
\]
for a constant $a>(1-\epsilon)/\epsilon$.  Then
\[
    \log N\le C(\log k)(\log M)=O\bigl((\log M)^{a+1}\bigr),
\]
and hence, for all sufficiently large $M$,
\[
    \log \gamma(N)=(\log N)^{1-\epsilon}
    \le O\bigl((\log M)^{(a+1)(1-\epsilon)}\bigr)
    <(\log M)^a\le \log k.
\]
Thus $\gamma(N)<k$.  Running the reduction and then the assumed polynomial-time
distinguisher would decide \textup{\textsc{3Sat}} in time
\[
    N^{O(1)}=2^{\log^{O(1)}M},
\]
contradicting
$\mathsf{NP}\nsubseteq \mathsf{DTIME}(2^{\log^{O(1)} n})$.

Finally, assume
\[
    \mathsf{NP}\nsubseteq\bigcap_{\delta>0}\mathsf{DTIME}(2^{n^\delta}).
\]
Since \textup{\textsc{3Sat}} is NP-complete under polynomial-time reductions,
there is a constant $\delta_0>0$ such that \textup{\textsc{3Sat}} is not in
$\mathsf{DTIME}(2^{M^{\delta_0}})$, where $M$ denotes the formula size.  Set
$\alpha:=\delta_0/2$ and choose
\[
    k=\left\lceil 2^{M^\alpha}\right\rceil .
\]
Then
\[
    \log N\le C M^\alpha\log M.
\]
For a sufficiently small constant $c>0$, depending only on $C$ and $\delta_0$,
we have, for all sufficiently large $M$,
\[
    \log\bigl(N^{c/\log\log N}\bigr)
    =\frac{c\log N}{\log\log N}
    < M^\alpha\le \log k.
\]
Thus $N^{c/\log\log N}<k$.  A polynomial-time distinguisher for this factor would
therefore decide \textup{\textsc{3Sat}} in time
\[
    N^{O(1)}=M^{O(\log k)}=2^{O(M^\alpha\log M)}\le 2^{M^{\delta_0}},
\]
contradicting the choice of $\delta_0$.
\end{proof}

\anyfieldinapproximability*

\begin{proof}
Fix a finite field $\F_q$, and let $M$ denote the input size of the Boolean
\textup{\textsc{QuadEq}} instance over $\F_q$.  By
\Cref{thm:elementary-moment-rank-gap}, there is an absolute constant $C$ such
that the reduction with rank parameter $k$ has matrix dimension and total output
size at most
\[
    N\le M^{Ck}.
\]
In the YES case the optimum is $1$, while in the NO case it is larger than $k$.
The source problem is NP-hard by \Cref{lem:quadeq-source-elementary}.

For a constant factor $\gamma>1$, choose a constant integer $k\ge\gamma$.  The
reduction is polynomial time, so a polynomial-time distinguisher for the gap
$1$ versus $>\gamma$ would imply $\mathsf{NP}=\mathsf{P}$.

Next fix $0<\epsilon<1$ and suppose there is a polynomial-time distinguisher for
\[
    \gamma(N)=(\log N)^{1-\epsilon}.
\]
Choose
\[
    k=\left\lceil(\log M)^a\right\rceil
\]
for a constant $a>(1-\epsilon)/\epsilon$.  Then
\[
    \log N\le Ck\log M=O\bigl((\log M)^{a+1}\bigr),
\]
and hence, for all sufficiently large $M$,
\[
    \gamma(N)
    = (\log N)^{1-\epsilon}
    \le O\bigl((\log M)^{(a+1)(1-\epsilon)}\bigr)
    < (\log M)^a\le k.
\]
Thus the distinguisher would solve Boolean \textup{\textsc{QuadEq}} over $\F_q$
in time
\[
    N^{O(1)}=M^{O(k)}=2^{\log^{O(1)}M},
\]
contradicting
$\mathsf{NP}\nsubseteq \mathsf{DTIME}(2^{\log^{O(1)} n})$.

Finally, assume
\[
    \mathsf{NP}\nsubseteq\bigcap_{\delta>0}\mathsf{DTIME}(2^{n^\delta}).
\]
Since Boolean \textup{\textsc{QuadEq}} over $\F_q$ is NP-hard under
polynomial-time reductions, there is a constant $\delta_0>0$ such that it is not
in $\mathsf{DTIME}(2^{M^{\delta_0}})$.  Set $\alpha:=\delta_0/2$ and choose
\[
    k=\left\lceil M^\alpha\right\rceil .
\]
Then
\[
    \log N\le C M^\alpha\log M.
\]
For a sufficiently small constant $c>0$, depending only on $C$ and $\delta_0$,
we have, for all sufficiently large $M$,
\[
    \frac{c\log N}{\log\log N}
    \le \frac{cC M^\alpha\log M}{\log(CM^\alpha\log M)}
    < M^\alpha\le k.
\]
Thus a polynomial-time distinguisher for the factor $c\log N/\log\log N$ would
solve Boolean \textup{\textsc{QuadEq}} over $\F_q$ in time
\[
    N^{O(1)}=M^{O(k)}=2^{O(M^\alpha\log M)}\le 2^{M^{\delta_0}},
\]
contradicting the choice of $\delta_0$.
\end{proof}

%% file: sec/proof_Khot_Saket.tex
\section{Proof of Lemmas from \cite{KhotSaket2014}}
\label{sec:proof-khot-saket}
\monomialassignmentsaspoints*
\begin{proof}
For each $a \in \F_2^{n+1}$, define the evaluation map $E_a : \F_2[x]_{\le d} \to \F_2$ by $E_a(q) = q(a)$. Then $E_a$ is in the dual space $\left(\F_2[x]_{\le d}\right)^*$. We claim that 
\[
\left(\F_2[x]_{\le d}\right)^*=\operatorname{span}\{E_a\}_{a \in \F_2^{n+1}}.
\]

To prove the claim, it suffices to show that if a polynomial $q \in \F_2[x]_{\le d}$ vanishes on all points in $\F_2^{n+1}$, then $q$ is the zero polynomial. Assume for contradiction that $q \neq 0$. Let $x^S$ be a monomial in $q$ with a non-zero coefficient $c_S = 1$ such that $S$ is \emph{minimal} with respect to set inclusion. Since $\deg(q) \le d$, we have $|S| \le d$. Let $\mathbbm{1}_S \in \F_2^{n+1}$ be the indicator vector of the set $S$. We now evaluate $q(\mathbbm{1}_S)$. A monomial $x^{S^{\prime}}$ evaluates to $1$ at $\mathbbm{1}_S$ if and only if $S^{\prime} \subseteq S$; otherwise, it evaluates to $0$. Thus,
\[
    q(\mathbbm{1}_S) = \sum_{S^{\prime} \subseteq S} c_{S^{\prime}}.
\]
By the minimality of $S$, all strict subsets $S^{\prime} \subsetneq S$ must have a coefficient $c_{S^{\prime}} = 0$ in $q$. Therefore, the sum collapses to $q(\mathbbm{1}_S) = c_S = 1 \neq 0$. This contradicts the assumption that $q$ vanishes on $\F_2^{n+1}$.

The extended assignment $\sigma$ is a linear functional on $\F_2[x]_{\le d}$, so it can be expressed as a linear combination of $\{E_a\}_{a \in \F_2^{n+1}}$. Since the field is $\F_2$, the coefficients must be $0$ or $1$. The subset $\beta \subseteq \mathbb{F}_2^{n+1}$ is simply the set of points $a$ whose corresponding functional $E_a$ has a coefficient of $1$, yielding $\sigma = \sum_{a \in \beta} E_a$. Therefore, $\sigma(q)=\sum_{a\in \beta} q(a)$ for every $q\in \mathbb{F}_2[x]_{\leq d}$.
\end{proof}

\symmetricdecomposition*

\begin{proof}
We argue by induction on $k=\operatorname{rank}(A)$.  The case $k=0$ is
immediate.  Assume $k>0$, and write $a_\ell$ for the $\ell$-th
column of $A$.

Suppose first that $A_{ii}=1$ for some $i$.  Put
$B:=A+a_ia_i^\top$.  The $i$-th column of $a_ia_i^\top$ is $a_i$, so the
$i$-th column of $B$ is zero. By symmetry, the $i$-th row of $B$ is also zero. Moreover, the column space of $B$ is
contained in the column space of $A$, and every column $a_\ell$ of $A$ can be
written as a column of $B$ plus a scalar multiple of $a_i$.  Hence
\[
    \operatorname{col}(A)=\operatorname{col}(B)+\operatorname{span}\{a_i\}.
\]
This sum is a direct sum, since every vector in $\operatorname{col}(B)$ has zero
$i$-th coordinate, whereas the $i$-th coordinate of $a_i$ is $1$.  Therefore
$\operatorname{rank}(B)=k-1$.  By induction, $B$ is a sum of at most
$\lfloor 3(k-1)/2\rfloor$ symmetric rank-one matrices of the form $uu^\top$.
Adding $a_ia_i^\top$ gives a decomposition of $A$ using at most
\[
    1+\lfloor 3(k-1)/2\rfloor \le \lfloor 3k/2\rfloor
\]
terms.

It remains to handle the situation in which all diagonal entries of $A$ are
zero.  Since $A$ is nonzero and symmetric, there are distinct indices $i,j$ with
$A_{ij}=A_{ji}=1$.  Put
\[
    B:=A+a_ia_j^\top+a_ja_i^\top .
\]
For this choice, the $i$-th and $j$-th columns of
$a_ia_j^\top+a_ja_i^\top$ are $a_i$ and $a_j$, respectively, and hence the
$i$-th and $j$-th rows and columns of $B$ are zero.  As above,
\[
    \operatorname{col}(A)=\operatorname{col}(B)+\operatorname{span}\{a_i,a_j\}.
\]
The sum is direct because every vector in $\operatorname{col}(B)$ vanishes on
coordinates $i$ and $j$, while the restrictions of $a_i$ and $a_j$ to these two
coordinates are $(0,1)^\top$ and $(1,0)^\top$.  Thus
$\operatorname{rank}(B)=k-2$.  Applying the induction hypothesis to $B$ and
using the identity, valid over $\F_2$,
\[
    a_ia_j^\top+a_ja_i^\top
    =a_ia_i^\top+a_ja_j^\top+(a_i+a_j)(a_i+a_j)^\top,
\]
we obtain a decomposition of $A$ with at most
\[
    3+\lfloor 3(k-2)/2\rfloor = \lfloor 3k/2\rfloor
\]
terms.  This completes the induction.
\end{proof}